\documentclass[10pt,journal,compsoc]{IEEEtran}

\usepackage[english]{babel}

\usepackage{amsmath}
\usepackage{amssymb}
\usepackage{amsthm}

\usepackage{algorithm}
\usepackage[noend]{algpseudocode} 
\usepackage{enumitem} 

\usepackage{graphicx}
\usepackage[colorlinks=true, allcolors=blue]{hyperref}
\usepackage{cleveref}
\usepackage{subcaption}
\usepackage{booktabs}

\usepackage{multirow}
\usepackage[normalem]{ulem}
\useunder{\uline}{\ul}{}

\newtheorem{definition}{Definition}
\newtheorem{theorem}{Theorem}
\newtheorem{proposition}{Proposition}
\theoremstyle{remark}
\newtheorem{example}{Example}
\newtheorem*{remark}{Remark}

\begin{document}
\bstctlcite{MyBSTcontrol}

\title{PADER: Paillier-based Secure Decentralized Social Recommendation}
\author{
Chaochao Chen, Jiaming Qian, Fei Zheng$^*$, Yachuan Liu

\IEEEcompsocitemizethanks{
\IEEEcompsocthanksitem C. Chen, F. Zheng, and Y. Liu are with the College of Computer Science and Technology, Zhejiang University, Hangzhou, China. \textit{Email: }zjuccc@zju.edu.cn, zfscgy2@zju.edu.cn, 3190105123@zju.edu.cn.
\IEEEcompsocthanksitem J. Qian is with the School of Software Technology, Zhejiang University, Ningbo, China. \textit{Email: }22451051@zju.edu.cn.
}

}

\markboth{Journal of \LaTeX\ Class Files,~Vol.~18, No.~9, September~2020}%
{How to Use the IEEEtran \LaTeX \ Templates}

\IEEEtitleabstractindextext{%
\begin{abstract}
The prevalence of recommendation systems also brings privacy concerns to both the users and the sellers, as centralized platforms collect as much data as possible from them.
To keep the data private, we propose PADER, a Paillier-based secure decentralized social recommendation system.
In this system, the users and the sellers are nodes in a decentralized network.
The training and inference of the recommendation model are carried out securely in a decentralized manner, without the involvement of a centralized platform.
To this end, we apply the Paillier cryptosystem to the SoReg (Social Regularization) model, which exploits both user's ratings and social relations.
We view the SoReg model as a two-party secure polynomial evaluation problem and observe that the simple bipartite computation may result in poor efficiency.
To improve efficiency, we design secure addition and multiplication protocols to support secure computation on any arithmetic circuit, along with an optimal data packing scheme that is suitable for the polynomial computations of real values.
Experiment results show that our method only takes about one second to iterate through one user with hundreds of ratings, and training with $\sim 500K$ ratings for one epoch only takes $<3$ hours, which shows that the method is practical in real applications. The code is available at \emph{\url{https://github.com/GarminQ/PADER}}.
\end{abstract}
}

\maketitle

\section{Introduction}
\IEEEPARstart{R}{ecommendation} systems are widely employed in various domains nowadays, e.g., shopping websites, social platforms, online media, and offline stores.
The mechanism of recommendation systems is to recommend new items (e.g., goods, movies) based on the user-item interaction (e.g., purchases, ratings) history and other side-informations of users and items (e.g., users' social data, description of items).
In the conventional recommendation approach, the e-commerce platform gathers data from both users and sellers, then trains the recommendation model in a centralized manner.
For example, prominent online shopping platforms such as Amazon or Taobao require the seller to post item information on the platform's website, while users must register on the platform to purchase goods.
Thus, all transactions between the seller and the user must go through the platform, and the interaction data is kept by the platform to train the recommendation model.
While the aforementioned centralized recommendation is efficient and has a high performance, it also causes privacy problems for both the user and the seller.
The data collected by the platform could be sensitive to the user and valuable to the seller.
However, in reality, users and sellers still have to compromise their private data in exchange for convenience.

The most essential data for the recommendation system is the user-item interaction, e.g., clicks, purchases, and ratings.
While some methods completely rely on the user-item interaction to make recommendations~\cite{2007pmf,koren2009mf,2016_wide_and_deep,hexiangnan2017ncf,2018din}, others also leverage side-information such as users' social relations~\cite{mahao2008sorec,jamali2010socialmf,mahao2011social_regularization,fanwenqing2019gnn_social} and the features of the item~\cite{weiyinwei2019mmgcn,zhaozhe2019rec_video,chengzhiyong2019mmalfm}, to further improve the recommendation performance.
However, on the other side, the excessive collection of data by recommendation systems further exacerbates the privacy problem.

Many existing studies have proposed solutions to tackle the privacy problem in recommendation.
One line of work applies cryptographic methods to matrix factorization or other relatively simple recommendation models and is usually based on a three-party setting, i.e., the users, the recommender (centralized platform in this case), 
and the crypto-service provider~\cite{nikoleanko2013ppmf,kim2018ppmf,fuanmimng2022ppmf,erkin2012privacy,badsha2017ppmf,chenjinrong2022secrec}.
These methods fully protect users' ratings and can keep the trained embeddings secret as long as the recommender and the crypto-service provider do not collude with each other.
Differential privacy~\cite{dwork2014dp}, which protects data privacy via random noises, is also applied to the privacy-preserving recommendation~\cite{berlioz2015dp_mf,zhuxue2016dp_rec,shenyilin2016epic_rec,megxuying2019privacy_social}, albeit it has to sacrifice more accuracy to maintain stronger privacy.
While the above privacy preservation methods are conducted in a centralized way, there are also some solutions for decentralized recommendation, i.e., making recommendations without a centralized platform.
These methods usually exchange the intermediate results~\cite{defiebre2020de_rec,defiebre2022de_rec,ccc2018decentralized_mf}  which potentially leak private data.
As for cryptography-based decentralized recommendation, existing work suffers from efficiency problems and only supports limited kinds of models~\cite{hoens2010private_social,chaidi2021federated_mf,tangqiang2018pp_friend_rec}.
Overall, the research of privacy-preserving decentralized recommendation is still in a relatively early stage.
Existing methods are either too computationally expensive or limited to simple models.

\begin{figure*}[h]
    \centering
    \begin{subfigure}{0.55\textwidth}
        \includegraphics[width=\textwidth]{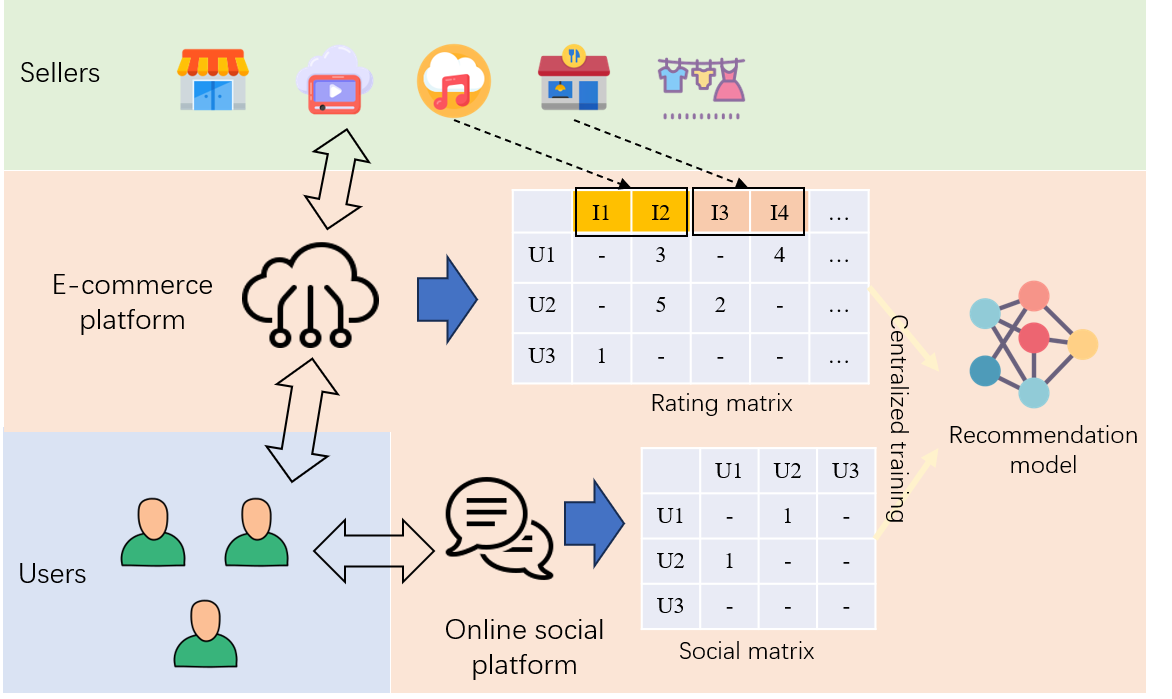}
        \caption{Centralized recommendation.}
    \end{subfigure}
    \begin{subfigure}{0.38\textwidth}
        \includegraphics[width=\textwidth]{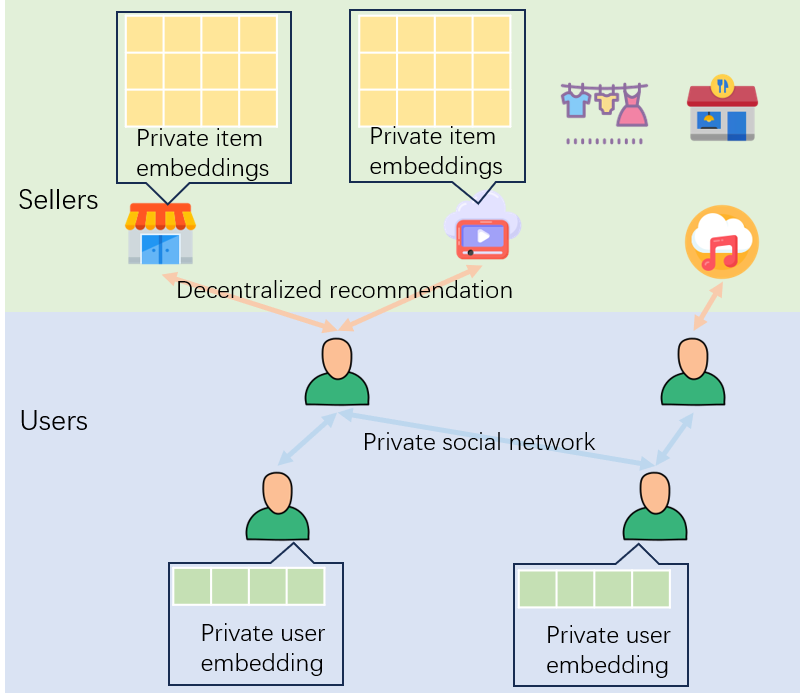}
        \caption{Decentralized recommendation.}
    \end{subfigure}
    \caption{Illustration of centralized and decentralized recommendation systems.}
    \label{fig:overview}
\end{figure*}

To achieve privacy-preserving and decentralization at the same time, we first design a practical decentralized recommendation scheme. 
In this scheme, we view all users and sellers (e.g., shops, restaurants, or online stores) as different nodes connected in the decentralized network.
The users keep their user profiles private, while the sellers keep their item profiles private.
The direct interaction between the user and the item (e.g., purchases) is known to both the user and the corresponding seller, while the user could keep the detailed interaction data and also rate the interaction privately.
The social relation data, e.g., users' friendships, are kept private by the user.
We provide a comparison of centralized recommendation and decentralized recommendation in \Cref{fig:overview}.
To make recommendations, we leverage both user-item interactions and users' social relations. 
We adopt a classical social recommendation model based on matrix factorization with social regularization~\cite{mahao2011social_regularization}.

To fully protect privacy, we use Additive Homomorphic Encryption (AHE), specifically, the Paillier cryptosystem, to train our model.
In the social regularization model, the embedding updates of the user and the item are polynomials of the original embeddings, the ratings, and the embeddings of the user's socially connected friends (neighbors).
Thus, it can be considered a secure two-party polynomial evaluation problem between the user and the seller (although the neighbors also participate in the computation, it is easy to incorporate them into the two-party computation seamlessly).
We notice that any polynomial can be decomposed into the bipartite form $\sum f_ig_i + f' + g'$, where $f_i, f'$ are terms that can be computed on one party and $g_i, g'$ can be computed on the other party. We call this a \textit{bipartite decomposition} of the polynomial.
Since AHE supports ciphertext-plaintext multiplication, the decomposed polynomial can be evaluated securely.

While it seems trivial to use AHE to securely train the social recommendation model via bipartite decomposition, we find that two important improvements can be made.
First, the bipartite decomposition may result in very bad efficiency as demonstrated in \Cref{exp:bipartite}.
This is because different arithmetic circuits representing the same polynomial may result in different computation efficiency.
To overcome this, we design secure addition and multiplication protocols based on AHE and random masking to support the computation of arbitrary circuits.
As a result, we can compute the model updates in a more `natural' order which greatly improves the efficiency compared with the naive bipartite computation.
Second, we notice that the data packing for AHE is feasible even when the computation is complicated.
Existing literature has not paid much attention to the data packing of AHE, as data packing is mostly used as a small trick to improve efficiency.
To unleash the full potential of data packing, in this paper, we study the data packing for AHE thoroughly while taking fixed-point arithmetic into consideration.
Moreover, we derived the optimal packing scheme by which we can pack as many plaintexts into one as possible while ensuring the circuit is evaluated correctly.
Combining these two improvements together, we propose a highly efficient secure computation protocol based on AHE for real-valued data.

In conclusion, we mainly make the following contributions:
\begin{itemize}
    \item We propose PADER, a decentralized social recommendation scheme based on Paillier AHE, which protects both the users' data and the recommenders' data, while making use of the user's private ratings and social relations.
    %
    \item We demonstrate secure two-party computation based on the bipartite decomposition and AHE may result in low efficiency, and propose a secure computation protocol that supports the computation of any arithmetic circuits with limited multiplicative depth.
    \item We study the data packing method for AHE thoroughly, and propose the optimal packing method for real-valued polynomial computation.
    \item We perform extensive experiments on synthetic and real-world datasets, showing that our proposed method is more efficient than the naive bipartite computation and fully homomorphic encryption in the decentralized social recommendation task.
\end{itemize}

\begin{table*}[ht]
\footnotesize
    \centering
    \begin{tabular}{cl}
    \toprule
    Notation        & Definition \\ \midrule
    $U, V$          & User and item embedding matrices. Each row ($U_i/V_j$) represents the embedding vector of one user/item. \\
    $R, r_{i,j}$             & Rating matrix and its entry in position $(i,j)$, represents the rating of $i$-th user on $j$-th item. \\
    $\mathcal R$    & Set of rated (user, item) pairs.\\
    $\mathcal N_i$  & Set of neighbors of the $i$-th user.\\
    $b_U, b_V$      & User bias and item bias.\\
    $L$             & Loss function. \\ 
    $\lambda_U, \lambda_V$  & Regularization coefficients on the user and item embeddings. \\
    $\lambda_S$             & Regularization coefficient for the social term in the SoReg model. \\
    \midrule
    $\text{Enc}, \text{Dec}$    & Encryption and decryption functions.\\
    $\mathbb Z_N^*$    & Set of integers less than $N$ and coprime to $N$.\\
    $\oplus$        & Homomorphic addition operator for ciphertext-ciphertext addition.\\
    $\otimes$       & Homomorphic multiplication operator for ciphertext-plaintext multiplication.\\
    $\mathsf{Encode}, \mathsf{Decode}$ & Encoding/decoding function between real numbers and integers. \\
    $N$             & Modulus for the encrypted computation, i.e., all computations are performed on $\mathbb Z_N$. \\
    $S$             & Scale factor for encoding/decoding.\\
    $\cdot^\text{(float)}, \cdot^\text{(fixed)}$ & Float and integer representations of the value. \\
    $G$             & Gates of a polynomial circuit, can be input gate, $\mathsf{Add}$, or $\mathsf{Mul}$. \\
    $W$             & Wires of each gate in a polynomial circuit. \\
    $\mathcal A$    & The party with the secret key in the secure computation, i.e., seller in our recommendation model.\\
    $\mathcal B$    & The party without the secret key in secure computation, i.e., user in our recommendation model.\\
    $v, c$          & Plaintext value and ciphertext value. \\
    $r$             & Random integer drawn from a uniform distribution from a certain integer ring.\\
    $\mathsf{Pack}, \mathsf{Unpack}$ & Packing and unpacking functions for data packing. \\
    $P$             & Slot size in data packing, i.e., the largest number - 1. \\
    $Q$             & Modulus used in data packing for each slot.\\
    \midrule
    $u_q$           & The $q$-th element in a user embedding vector. \\
    $u_p$           & The element contained at the new index $p$ in the user embedding vector after the vector have been packed. \\
    $v_{i,q}$       & The $q$-th element in $i$-th item's embedding vector. \\
    $v_{i,p}$       & The elements contained at the new index $p$ in $i$-th item's embedding vector after the vector have been packed. \\
    $w_i$           & The weight of $j$-th item of a certain user. E.g., rated = 1, unrated = 0.\\
    $f_{j,p}$       & The $p$-th element in the $j$-th neighbor's embedding. \\
    $n$             & Number of rated items of a certain user and a certain seller. \\
    $k$             & Dimension of embeddings. \\
    $m$             & Number of selected neighbors of a certain user in a training round. \\
    \bottomrule
    \end{tabular}
    \caption{Notations used in this paper.
    }
    \label{tab:notations}
\end{table*}
\section{Related Work}
\subsection{Decentralized and Privacy-Preserving Recommendation}
Decentralized recommendation distributes tasks such as training and inference across multiple clients (e.g., user devices) instead of relying on a centralized server. This approach helps protect data privacy by eliminating the need for users and sellers to upload their data to a central location.
A straightforward solution involves clients training the entire recommendation model using federated learning.
Current research mostly focuses on the Matrix Factorization (MF)~\cite{koren2009mf} model regarding the computation power of user devices.
Hegeds et al. \cite{fed_gossip_2019} compared federated learning and gossip learning (where there is no master node) on the decentralized MF model.
Ling et al. \cite{lingqing2012dec_mc} adopt a faster optimization method for decentralized MF.
Duriakova et al. \cite{duriakova2019pdmf} conduct item embedding aggregation from users' neighbors.
Chen et al. \cite{ccc2018decentralized_mf} propose to protect privacy by dividing the item embeddings into the global part and local part.
Chai et al. \cite{chaidi2021federated_mf} use a server to aggregate the MF model securely via homomorphic encryption, at the price of heavy computation cost.
Recommendation approaches other than MF are also explored.
Defiebre et al. \cite{defiebre2020de_rec,defiebre2022de_rec} make recommendations by querying similar users.
Long et al. \cite{long2023decentralized} aggregate device-side models by geologically and semantically similar users.
Hoens et al. \cite{hoens2010private_social} use Paillier cryptosystem to compute the average ratings of user's neighbors to make recommendation.

There are also many privacy-preserving recommendation methods that are not based on the decentralized setting.
Instead, the users upload their ratings to the server in a secure manner.
Nevertheless, they are still quite related to our approach.
Nikolaenko et al. \cite{nikoleanko2013ppmf} designed a garbled circuit for matrix factorization.
Kim et al. \cite{kim2018ppmf,fuanmimng2022ppmf} conduct secure matrix factorization by two-party computation based on homomorphic encryption.
Duan et al. \cite{duanjia2021secure_mf} propose an encryption scheme based on random permutation for outsourced MF.
Erkin et al. \cite{erkin2012he_rec,erkin2012privacy} perform similarity-based recommendation by homomorphic encryption.
To summarize, although there are multiple secure recommendation methods, most of them still focus on the centralized case.
For decentralized recommendation, studies are either not achieving cryptographic security or limited to relatively simple methods like similarity-based recommendation.

\subsection{Secure Computation}
To date, there are many methods to perform secure computation with two or multiple parties.
\textit{Homomorphic Encryption} (HE) provides a solution by allowing arithmetic operations to be performed on the ciphertext.
HE schemes are usually classified into two classes according to supported operations.
\textit{Additive Homomorphic Encryption} (AHE)~\cite{paillier1999} allows addition on the ciphertexts, and consequently, multiplication between ciphertext and plaintext is supported.
\textit{Fully Homomorphic Encryption} (FHE)~\cite{gentry2009,gentry2013,bgv2014,ckks2017}, invented recently, supports both addition and multiplication on the integer or polynomial ring.
Garbled Circuits (GC)~\cite{yao1986gc}, on the other side, can evaluate any boolean circuits and are mainly used for non-linear functions.
While HE and GC are specific techniques for secure computation, general MultiParty Computation (MPC) protocols~\cite{aby2015,secureml2017,ezpc2020,spdz2020} usually mix multiple techniques to support more computations. 
They usually use secret sharing~\cite{shamir1979share} or HE for addition and multiplication, while using GCs or other customized protocols for non-linear functions.
As for secure polynomial evaluation, current methods are mostly developed under a multi-server setting. 
They decompose the polynomials into the sum of multiplicative terms and then compute them separately~\cite{hoss2022polynomial,gajera2019polynomial}.

Existing work on secure computation mostly focuses on a two-party setting, and the computation is divided into the offline phase (e.g., generation of Beaver triples) and the online phase (e.g., secret-shared computation).
In our paper, we focus on the decentralized setting instead.
We do not consider the offline phase here since the participants are not fixed.
\section{Preliminaries}
\subsection{Social Regularization Model}
\label{sec:mf-def}
Matrix factorization~\cite{2007pmf,koren2009mf} is a classic method for recommendation.
Suppose we have a matrix $R \in \mathbb R^{N_U \times N_I}$, where each entry $r_{i, j}$ is the rating or interaction of the $i$-th user to the $j$-th item, and $N_U, N_I$ denote the number of users and items.
Matrix factorization is to find two low rank matrices $U\in \mathbb R^{N_U\times k}$ and $V \in \mathbb R^{N_I\times k}$ such that $k \ll N_U, N_I$, and $UV^T \approx R$.
Here, $U$ and $V$ are called user embeddings and item embeddings respectively.
While there are many variants of matrix factorization, we consider the most widely-used one whose loss is usually expressed as follows:
\begin{equation}
\small
\label{eq:mf-loss}
L = \sum_{i, j \in \mathcal R} (U_iV_j^T + {b_U}_i + {b_I}_j - r_{i, j})^2 + \lambda_U \Vert U_i \Vert_2^2 + \lambda_V \Vert V_j\Vert_2^2.
\end{equation}
Where $\mathcal R$ denotes the set of rated (user, item) pairs. $b_U$ and $b_I$ are the bias vectors of the users and items.
For example, if the $j$-th item is highly rated among most people, ${b_I}_j$ should be large.
$\lambda_U$ and $\lambda_V$ are regularization coefficients used to control overfitting.

As the above model only considers the interaction between users and items, social recommendation~\cite{mahao2008sorec,jamali2010socialmf,mahao2011social_regularization} also leverages the social relationships between users.
The intuition behind the social recommendation is that if two people are friends, they probably have similar preferences.
Consequently, their embedding vectors tend to be similar.
To capture this phenomenon, in this paper, we consider the SoReg (Social Regularization)~\cite{mahao2011social_regularization} model which adds a social term to the original matrix factorization loss \eqref{eq:mf-loss}:
\begin{equation}
\label{eq:smf-loss}
\small
\begin{split}
    L = & \sum_{i, j \in \mathcal R} (U_iV_j^T + {b_U}_i + {b_I}_j - r_{i, j})^2 + \lambda_U \Vert U_i \Vert_2^2 
    \\& + \lambda_V \Vert V_j\Vert_2^2 + \underbrace{\dfrac{\lambda_S}{|\mathcal N_i|}\sum_{k\in\mathcal N_i} \Vert U_i - U_k\Vert_2^2}_\text{social term},
\end{split}
\end{equation}
where $\mathcal N_i$ denotes the set of friends of user $i$, and $\lambda_S$ is to control the importance of the social term during training.

\subsection{Homomorphic Encryption}
Homomorphic Encryption (HE)~\cite{paillier1999,gentry2009} allows certain computations on the ciphertexts so that we can analyze data while protecting data privacy.
It can be categorized as either fully HE or partial HE.
Fully HE allows both multiplication and addition on the ciphertext, while partial HE only allows one type of operation.
In this paper, we utilize the partial HE that supports addition, which is also referred to as Additive Homomorphic Encryption (AHE).
AHE satisfies
\begin{equation}
    \text{Enc}(x + y) = \text{Enc}(x) \oplus \text{Enc}(y),
\end{equation}
where $x, y$ are plaintexts, $\text{Enc}$ is the encryption function, and $\oplus$ denotes the homomorphic addition on the ciphertext.
It is worth noting that the operation on the ciphertext also requires the public key.
One important property of Additive HE is that it also supports the multiplication between plaintext and ciphertext.
This is because 
\begin{equation}
\small
    \text{Enc}(\underbrace{x + \cdots + x}_\text{$y$ $x$'s}) = \underbrace{\text{Enc}(x) \oplus \cdots \oplus \text{Enc}(x)}_\text{$y$ times} = \text{Enc}(x)^{\oplus y}.
\end{equation}
For convenience, we use $\text{Enc}(x)\otimes y$ to denote $\text{Enc}(x)^{\oplus y}$ (homomorphic addition on $\text{Enc}(x)$ for $y$ times), since it corresponds to multiplication on the plaintext.

In this paper, we use the classical Paillier cryptosystem~\cite{paillier1999} for AHE.
Paillier cryptosystem takes an integer $ 0 \le x < N$ and a random factor $0 < r < N$ coprime to $n$ (denoted as $r\stackrel{\$}\leftarrow \mathbb Z_N^*$) as input, while $N = pq$ is the product of two large prime numbers.
Thus, the encryption function is $\text{Enc}(x;r)$, while we omit the random factor $r$ in most parts of the paper for convenience.
Because of the random factor, the Paillier cryptosystem is not deterministic, meaning that one plaintext can have multiple ciphertexts.
Moreover, the Paillier cryptosystem is IND-CPA secure, i.e., given the plaintext $v$ and two ciphertexts $c, c'$ with one of them being the encryption of $v$, any probabilistic polynomial-time adversary cannot guess the correct encryption with a probability of larger than $1/2 + \mu(N)$ where $\mu(N)$ is a negligible function of the security parameter.
We omit the details of the implementation of the Paillier cryptosystem as it is widely described in the existing literature.
\section{Secure Computation based on AHE}
\label{sec:two-party-computation}
In this section, we consider the secure two-party polynomial computation based on Paillier AHE.
We describe the problem as follows:
Suppose two parties $\mathcal A$ and $\mathcal B$ want to collaboratively compute the multivariate polynomial 
\begin{equation}
f(a_1, \cdots, a_m, b_1, \cdots, b_n) = f(\mathbf a, \mathbf b),
\end{equation}
where $\mathbf a = (a_1, \cdots, a_m) \in \mathbb R^m$ are $\mathcal A$'s private inputs, and $\mathbf b = (b_1, \cdots, b_n)\in\mathbb R^n$ are $\mathcal B$'s private inputs.
After computation, only the designated party (can be either $\mathcal A$, $\mathcal B$, or both) gets the result, while all other intermediate values and private inputs are kept confidential during the whole computation process.

To solve this secure computation problem using AHE, we first discuss the conversion from real numbers to integers, as a prior step for encrypted computation.
Then we propose a simple bipartite decomposition method for two-party computation and show that the efficiency can be improved by changing the computation circuit.
Based on this observation, we propose a two-party computation protocol that can compute arbitrary arithmetic circuits.
Moreover, we design an optimal data packing scheme to further improve efficiency.

\subsection{Fixed-Point Arithmetics}
\label{sec:fixed-point}
Paillier's AHE, like many other cryptosystems, only takes elements in an integer ring $0 \le x < N$ as input.
On the contrary, most computation is performed on real numbers in real-world applications.
To convert real values (here we use the `float' superscript to denote real numbers, as they are usually represented in the float format) to integer values, we have to use the fixed-point representation, which is to scale the real number and only keep the integer part, i.e. $x^\text{(fixed)} = \lfloor x^\text{(float)}\cdot S \rceil$, where the \textit{scale factor} $S$ is a large integer.
As for negative numbers, a conventional way is to use $N - x^\text{(fixed)}$ to represent $-x^\text{(float)}$.
Hence, we can write the conversion from real number to integer as follows:
\begin{equation}
\small
\begin{split}
\label{eq:encode}
    \mathsf{Encode}(x^{\text{(float)}}; N, S) = 
    \begin{cases}
        \left\lfloor x^{\text{(float)}} \cdot S \right\rceil & \text{if } x^{\text{(float)}} < \frac{N}{2S}, \\
        N - \left\lfloor x^{\text{(float)}} \cdot S \right\rceil & \text{otherwise}.
    \end{cases}
\end{split}
\end{equation}

When converting a fixed-point number to its original float-point value, we first decide its sign by checking whether $x^\text{(fixed)} < N/2$ or not, then scale them using $1/S$:
\begin{equation}
\small
\begin{split}
\label{eq:decode}
    \mathsf{Decode}(x^{\text{(fixed)}}; N, S) = 
    \begin{cases}
        x^{\text{(fixed)}} / S & \text{if } x^{\text{(fixed)}} < \frac{N}{2}, \\
        (N - x^{\text{(fixed)}}) / S & \text{otherwise}.
    \end{cases}
\end{split}
\end{equation}

By the above conversions, we can see that fixed-point computation supports addition and multiplication modulo $N/S$. 
For two arbitrary real numbers $x, y$ such that $|x|, |y| \ll N/S$ and ignoring the small rounding error, we have:
\begin{equation}
\small
    \begin{cases}
    \mathsf{Decode}(\mathsf{Encode}(x) + \mathsf{Encode}(y)) = x + y, \\
    \mathsf{Decode}(\mathsf{Encode}(x) \cdot \mathsf{Encode}(y)) / S = x \cdot y,
    \end{cases}
\end{equation}
as there will be no overflow/underflow during the computation.

\subsection{Bipartite Decomposition}
\begin{definition}
    A bipartite decomposition for polynomial $f(\mathbf a, \mathbf b)$ is to decompose $f$ in the form
    \begin{equation}
        f(\mathbf a, \mathbf b) = \sum_{i=1}^k h_{\mathcal A, i}(\mathbf a)h_{\mathcal B, i}(\mathbf b) + g_{\mathcal A}(\mathbf a) + g_{\mathcal B}(\mathbf b),
    \end{equation}
    where $k$ represents the number of product terms in the decomposition. $\{h_{\mathcal A,i}(\mathbf a)\}, g_{\mathcal A}(\mathbf a)$ are polynomials taking only $\mathcal A$'s inputs, while $\{h_{\mathcal B,i}(\mathbf b)\}, g_{\mathcal B}(\mathbf b)$ are polynomials taking only $\mathcal B$'s inputs.
\end{definition}
Obviously, any multivariate polynomial taking two parties' inputs can be decomposed into such bipartite form since it can be written as the sum of monomials (product of variables).
Suppose that $\mathcal A$ has the secret key for AHE.
Leveraging the AHE's ability on ciphertext-plaintext multiplication, we can securely evaluate the polynomial in its bipartite form as follows:

\begin{equation}
\small
\begin{split}
    \text{Enc}[f(\mathbf a, \mathbf b)] = \oplus_{i=1}^k \left(\text{Enc}[h_{\mathcal A, i}(\mathbf a)]\otimes h_{\mathcal B, i}(\mathbf b)\right) \\ 
    \oplus \text{Enc}[g_{\mathcal A}(\mathbf a)] \oplus \text{Enc}[g_{\mathcal B}(\mathbf b)].
\end{split}
\end{equation}
We formalize the algorithm for secure computation by bipartite decomposition in \Cref{alg:bipartite}.

\begin{algorithm}[ht]
\begin{algorithmic}[1]
\footnotesize
\caption{Bipartite secure computation}
\label{alg:bipartite}
\Procedure{SecurebipartiteCompute}{$f(\mathbf a, \mathbf b) = \sum h_{\mathcal A, i}(\mathbf a)h_{\mathcal B, i}(\mathbf b) + g_{\mathcal A}(\mathbf a) + g_{\mathcal B}(\mathbf b)$}
    \State $\mathcal{A}$ computes all its local terms $\{ h_{\mathcal A, i}(\mathbf a) \}, g_{\mathcal A}(\mathbf a)$, then convert them into fixed-point representations via \cref{eq:encode}: $H_{\mathcal A, i} = [h_{\mathcal A, i}(\mathbf a)]^\text{(fixed)}, G_{\mathcal A} = [g_{\mathcal A}(\mathbf a)]^\text{(fixed)}$.
    \State $\mathcal A$ computes $\{ \text{Enc}(H_{\mathcal A, i}) \}, \text{Enc}(G_{\mathcal A}\cdot S)$ and sends them to $\mathcal B$.
    \State $\mathcal B$ computes $c = \oplus_{i=1}^k \left(\text{Enc}(H_{\mathcal A, i}) \otimes H_{\mathcal B, i}\right) \oplus \text{Enc}(G_{\mathcal A} \cdot S) \oplus \text{Enc}(G_{\mathcal B} \cdot S)$.
    

    \If{Only A is supposed to know the result}
        \State $\mathcal B$ sends $c$ to $\mathcal A$;
        \State $\mathcal A$ obtains the final result $v = \mathsf{Decode}(\text{Dec}(c)/S)$;
    \ElsIf{Only B is supposed to know the result}
        \State $\mathcal B$ adds a random number $r \stackrel{\$}{\leftarrow} \mathbb Z_N$ to the encrypted result, and sends $c \oplus \text{Enc}(r)$ to $\mathcal A$;
        \State $\mathcal A$ decrypts it and sends $\text{Dec}(c) + r$ to $\mathcal B$;
        \State $\mathcal B$ gets the final result $v = \mathsf{Decode}(\text{Dec}(c)/S)$;
    \ElsIf{Both A and B are supposed to know the result}
        \State $\mathcal B$ sends $c$ to $\mathcal A$;
        \State $\mathcal A$ obtains the final result $v = \mathsf{Decode}(\text{Dec}(c)/S)$;
        \State $\mathcal A$ shares $v$ with $\mathcal B$;
    \EndIf
    
\EndProcedure
\end{algorithmic}
\end{algorithm}

\subsubsection{Inefficiency of bipartite Computation}
Although bipartite decomposition with AHE provides a simple way to securely evaluate polynomials between two parties, it can be less efficient on certain problems.
Now we give an illustrative example to show its inefficiency.
\begin{example}
\label{exp:bipartite}
Consider the polynomial 
\begin{equation}
f(\mathbf a, \mathbf b, \mathbf x, \mathbf y) = 
\sum_{i=1}^n{a_ix_i}\sum_{i=1}^n{b_iy_i} = \sum_{i, j = 1}^n a_ib_jx_iy_j,
\end{equation}
where $\mathbf a, \mathbf b$ are owned by $\mathcal A$, while $\mathbf x, \mathbf y$ are owned by $\mathcal B$.
The bipartite decomposition of $f$ has $n^2$ terms, i.e., for each pair of $i, j$, the two parties compute the product $(a_ib_j)(x_iy_j)$, then sum them up.
This requires $n^2$ encryptions and $\otimes$'s.
On the other side, one can also first compute $A = \sum a_ix_i$ and $B = \sum b_iy_i$, which only requires (roughly) $2n$ times of encryptions and $\otimes$'s.
However, this brings new difficulty of computing the product of $A$ and $B$ securely, which we will discuss later.
\end{example}
\begin{remark}
One may wonder whether $f$ has another bipartite decomposition with less than $n^2$ terms.
Unfortunately, the answer is no.
Consider the 4-tensor $T$ representing the coefficient of the term $a_ib_jx_ky_l$.
We can see that 
\begin{equation}
T_{i, j, k, l} = \begin{cases}
    1 & \text{if $i=k$ and $j=l$} \\
    0 & \text{otherwise}
\end{cases}
\end{equation}
We then convert $T$ into a $n^2\times n^2$ matrix by flattening its dimensions 1,2 and 3,4 respectively.
Interestingly, it becomes the identity matrix $I_{n^2} \in \mathbb R^{n^2\times n^2}$.
Then any bipartite decomposition can be viewed as a decomposition of $I_{n^2}$, such that $\sum_{r=1}^R p_r q_r^T = I_n^2$, where $p_r, q_r \in \mathbb R^{n^2}$ whose $(i\cdot n + j)$-th element represents the coefficients of $a_ib_j$ or $x_iy_j$.
While the rank of the identity matrix is the same as its dimension, any bipartite decomposition has at least $n^2$ terms.

The above example demonstrates that, although bipartite decomposition enables secure two-party computation of any polynomials, there could be more efficient ways to compute the same polynomial.
Next, we will design a protocol that supports the polynomial computation of any order.
\end{remark}

\subsection{Arithmetic Circuit}
Here we formally define the circuit representation for polynomials.

\begin{definition}[Arithmetic circuit of a polynomial]
For a polynomial $f(x_1, \cdots, x_m)$, an arithmetic circuit for computing it can be represented as $(G, W)$.
$G = (g_1, g_2, \cdots, g_{|G|})$ is the list of gates (e.g., $\mathsf{Add}$, $\mathsf{Mul}$, or input gate).
$W = [(w_{1,1}, w_{1,2}, \cdots), \cdots, (w_{|W|,1}, w_{|W|, 2}, \cdots)]$ are the indices of gates' input wires, where $w_{j, 1}, w_{j, 2}, \cdots$ is the indices of the input wires of $G[j]$.
For example, $G[j]$'s inputs are $G[w_{j,1}]$ and $G[w_{j, 2}]$'s outputs (if it has two input wires).
Notice that, a gate can contain an arbitrary number of input wires.
For example. if $G[j]$ is an input variable, then $G[j] = ()$ is an empty tuple since no input wires are needed, and $\mathsf{Add}, \mathsf{Mul}$ gates have two input wires.
To evaluate the polynomial in order, $G$ must satisfy that for any $j \in [1, |W|]$, $w_{j, 1}, w_{j, 2} < j$, i.e., before computing the output of one gate, its inputs must be computed first.
\end{definition}
\begin{example}
Consider the polynomial $f(x_1, x_2, x_3) = x_1(x_2 + x_3)$, we can write it as
\begin{equation}
        f(x_1, x_2, x_3) = \mathsf{Mul}(x_1, \mathsf{Add}(x_2, x_3)).
    \end{equation}
\end{example}
Hence, one arithmetic circuit of $f$ is 
\begin{equation}
\begin{cases}
    G = (x_1, x_2, x_3, \mathsf{Add}, \mathsf{Mul}), \\
    W = [(), (), (), (2, 3), (1, 4)].
\end{cases}
\end{equation}

\subsubsection{Local Computation}
Since we are only interested in the secure computation part, we hope to pre-compute any variables that can be locally computed on one party.
When a gate's inputs all belong to the same party, we can compute them locally without involving any interactive computation between two parties.
By pre-computing those local results, we can transform a circuit into a more \textit{compact} form.
For example, consider $f(a, b, x, y) = ab + xy$, where $a, b$ belong to $\mathcal A$ and $x, y$ belong to $B$.
The local computations can be represented as $c = ab$ (computed on $\mathcal A$) and $z = xy$ (computed on $\mathcal B$),
yielding the compact circuit $f'(c, z) = c + z$.
We describe this algorithm formally in \Cref{alg:local-compute}.

\begin{algorithm}[ht]
\begin{algorithmic}[1]
\small
\caption{Local computation to obtain a compact circuit}
\label{alg:local-compute}
\Procedure{LocalCompute}{$G, W$}
\State LocalVars$_\mathcal A$ $\gets [\ ]$
\State LocalVars$_\mathcal B$ $\gets [\ ]$
\For{$i = 1$}{$|G|$}
    \If{$G[i]$ is a input of $\mathcal C \in \{\mathcal A, \mathcal B\}$} 
        \State Add $i$ to LocalVars$_{\mathcal C}$
    \ElsIf{$\forall j \in W[i], j$ in LocalVars$_{\mathcal C }$, $\mathcal C \in \{\mathcal A, \mathcal B\}$}
        \State Add $i$ to LocalVars$_{\mathcal C}$
        \State Replace $G[i]$ with an expression that can be computed on party $\mathcal C$ locally.
    \EndIf
\EndFor
\For {$i = 1$}{$|G|$}
    \State ParentNodes $\gets$ find all $j$ such that $i \in W[j]$
    \If{ParentNodes $\subseteq$ LocalVars$_{\mathcal A} \cup$ LocalVars$_{\mathcal B}$}
        \State Delete the $i$-th element of $G$ and $W$
        \State $\forall i \in \{1, 2, \cdots, |W|\}$ replace $j$ with $j - 1$ for $j \in W[i]$ and $j \ge i$.
    \EndIf
\EndFor
\State \Return $(G, W)$
\EndProcedure
\end{algorithmic}
\end{algorithm}

\subsubsection{Level of the Variable}
The arithmetic circuit representation of a polynomial is agnostic to the domain of the polynomial.
However, in our case, we want to convert a real-valued polynomial into an integer one.
Due to the scale factor $S$ multiplied for each input, each multiplication amplifies the result by $S$.
Therefore, we define the \textit{level} of the variable as the number of $S$ multiplied.
Original inputs have a level of $1$ since they are only multiplied by $S$ once.
When adding two variables, the variable of the lower level shall be multiplied by the powers of $S$ to reach the higher level.
To better understand this, we consider a simple polynomial $f(x) = ab + c$.
If we just convert each input variable to its fixed-point representation, we will have $f'(x) = abS^2 + cS$, which is meaningless.
On the other side, by `aligning' their levels during addition, we have $f'(x) = abS^2 + cS^2 = f(x)S^2$, thus we can divide the outcome by $S^2$ to retain the actual result.
Hence, before performing the two-party secure computation in the integer domain, we must compute the levels of all gate outputs.
We formally describe the level computation algorithm in \Cref{alg:level-compute}.
\begin{algorithm}[ht]
\begin{algorithmic}[1]
\small
\caption{Compute levels of variable}
\label{alg:level-compute}
\Procedure{ComputeLevels}{$G, W$}
\State Levels $\gets [\ ]$
\For{$i = 1$}{$|G|$}
    \If{$G[i]$ is a variable or expression}
        \State Levels$[i]\gets 1$
    \ElsIf{$G[i] \text{ is } \mathsf{Add}$}
        \State Levels$[i] \gets \max_{j\in W[i]}\{\text{Levels}[j]\}$
    \ElsIf{$G[i] \text{ is } \mathsf{Mul}$}
        \State Levels$[i] \gets \sum_{j\in W[i]} \text{Levels}[j]$
    \EndIf
\EndFor
\State \Return Levels
\EndProcedure
\end{algorithmic}
\end{algorithm}

\subsubsection{Two-party Secure Addition and Multiplication}
By now, we have performed the local computation and obtained a compact polynomial, as all local computation is finished.
We have also computed the level of each intermediate variable.
Based on these, we can perform two-party secure addition and multiplication to finally produce the result of the polynomial.
Our method is motivated by \cite{KSS2013modular_design}, which has provided a high-level description of secure function evaluation based on AHE.

To conduct secure computation, we must ensure that the actual values of intermediate variables remain secret to both parties.
Hence, we stick to the following principles:
\begin{itemize}
    \item $\mathcal A$ can only have the masked value of any intermediate variable since it has the secret key.
    Specifically, for any intermediate variable $v$, $A$ can only get $v + m \bmod N$, where $m \stackrel{\$}{\leftarrow} \mathbb Z_N$ is randomly drawn from $\mathbb Z_N$ with equal probability and is not known to $\mathcal A$.
    \item $\mathcal B$ can only have the ciphertext of the intermediate variables, i.e., for any intermediate value $v$, $\mathcal B$ can only get $\text{Enc}(v)$.
\end{itemize}

Thus, both $\mathcal A$ and $\mathcal B$ cannot have any knowledge about the intermediate variables, in information-theoretic and computational senses, respectively.
Based on this, we can design a protocol for two-party computation on arithmetic circuits, by implementing secure addition and multiplication for all kinds of variables.
There are four cases of addition and multiplication regarding the types of input variables (we omit symmetric cases).
Note two input variables as $x$ and $y$, the cases are: 
\begin{itemize}[itemsep=1pt,parsep=0pt]
    \item $x$ belongs to $\mathcal A$ and $y$ belongs to $\mathcal B$ (denoted as $\mathcal{AB}$); 
    \item $x$ is an intermediate variable and $y$ belongs to $\mathcal A$ (denoted as $\mathcal{IA}$);
    \item $x$ is an intermediate variable and $y$ belongs to $\mathcal B$ (denoted as $\mathcal{IB}$);
    \item  $x, y$ both are intermediate variables (denoted as $\mathcal{II}$).
\end{itemize}
Based on this, We design $\mathsf{SecureAdd}$ and $\mathsf{SecureMul}$ algorithms for all these four cases.
We assume that each party has a storage module to store variables in key-value pairs.
Thus, each variable is represented by an integer (key), while $\mathcal A$ or $\mathcal B$ could find the corresponding plaintext or ciphertext by taking this integer as a key.
Specifically, each operator function takes the following arguments:
$i_1, i_2, j$ are the keys of the first input, the second input, and the output;
level$_1$, level$_2$, and opLevel are the levels of the first input, the second input, and the output.

For the addition of two variables, we can simply let $\mathcal B$ perform homomorphic addition on ciphertexts and get the result as an intermediate variable.
Note that we must align the levels of the variables before addition.
Multiplication is a little more complicated.
If one operand belongs to $\mathcal B$, simply using $\otimes$ yields the encrypted product on $\mathcal B$.
If the operands are one intermediate variable and one variable of $\mathcal A$'s, or both are intermediate variables, we have to use random masks (i.e., additive blinding in \cite{KSS2013modular_design}) to assist computation.
For example, consider $\text{Enc}(v_1), \text{Enc}(v_2)$ are two intermediate values held by $\mathcal B$.
To multiply them, $\mathcal B$ first generates two random masks $r_1, r_2$, and then sends $\text{Enc}(v_1 - r_1), \text{Enc}(v_2 - r_2)$ to $\mathcal A$.
As $\mathcal A$ has the secret key, he can easily compute $\text{Enc}[(v_1 + r_1)(v_2 +r_2)]$ and sends it back to $B$.
Notice that $v_1v_2 = (v_1 + r_1)(v_2 + r_2) - v_1r_2 - v_2r_1 - r_1r_2$, 
the ciphertext of the first term is already computed by $\mathcal A$, 
while the ciphertext of the latter three terms can also be computed locally on $\mathcal B$.
So $\mathcal B$ can obtain the result $\text{Enc}(v_1v_2)$.
We formally describe the secure multiplication protocol in \Cref{alg:smul}.

\begin{algorithm}[h!]
\begin{algorithmic}[1]
\footnotesize
\caption{Secure addition of two variables}
\label{alg:sadd}
\Procedure{SecureADD$_{\mathcal{AB}}$}{$i_1, i_1, j$, level$_1$, level$_2$, opLevel}
    \State $\mathcal A$ finds $v_1$ with key $i_1$
    \State $\mathcal B$ finds $v_2$ with key $i_2$
    \State $\mathcal A$ computes $v_1 \gets v_1 \cdot S^{\text{opLevel} - \text{level}_1}$
    \State $\mathcal B$ computes $v_2 \gets v_2 \cdot S^{\text{opLevel} - \text{level}_2}$
    \State $\mathcal A$ encrypts $v_1$ and sends $\text{Enc}(v_1)$ to $B$
    \State $\mathcal B$ computes $\text{Enc}(v_1 + v_2) = \text{Enc}(v_1) \oplus \text{Enc}(v_2)$ and stores the result with key $j$
\EndProcedure

\Procedure{SecureADD$_{\mathcal{IA}}$}{$i_1, i_2, j$, level$_1$, level$_2$, opLevel}
    \State $\mathcal B$ finds $c_1$ with key $i_1$
    \State $\mathcal A$ finds $v_2$ with key $i_2$
    \State{// Notice that opLevel = level$_1$ since the intermediate variable's level is higher than the input}
    \State $\mathcal A$ computes $v_2 \gets v_2 \cdot S^{\text{opLevel} - \text{level}_2}$ 
    \State $\mathcal A$ encrypts $v_2$ and sends $\text{Enc}(v_2)$ to $B$
    \State $\mathcal B$ computes $c_1 \oplus \text{Enc}(v_2)$ and stores the result with key $j$.
\EndProcedure

\Procedure{SecureADD$_{\mathcal{IB}}$}{$i_1, i_2, j$, level$_1$, level$_2$, opLevel}
    \State $\mathcal B$ finds $c_1$ with key $i_1$
    \State $\mathcal B$ finds $v_2$ with key $i_2$
    \State $\mathcal B$ computes $v_2 \gets v_2 \cdot S^{\text{opLevel} - \text{level}_2}$
    \State $\mathcal B$ computes $c_1 \oplus \text{Enc}(v_2)$ and stores result with key $j$.
\EndProcedure

\Procedure{SecureADD$_{\mathcal{II}}$}{$i_1, i_2, j$, level$_1$, level$_2$, opLevel}
    \State $\mathcal B$ finds $c_1, c_2$ with key $i_1, i_2$    
    \If{level$_1 <$ opLevel}
        \State $\mathcal B$ computes $c_1 \gets c_1 \otimes S^\text{opLevel - level$_1$}$
    \EndIf
    \If{level$_2 <$ opLevel}
        \State $\mathcal B$ computes $c_2 \gets c_2 \otimes S^\text{opLevel - level$_2$}$
    \EndIf
    \State $\mathcal B$ computes $c_1 \oplus c_2$ and stores the result with key $j$.
\EndProcedure
\end{algorithmic}
\end{algorithm}

\begin{algorithm}[h!]
\begin{algorithmic}[1]
\footnotesize
\caption{Secure multiplication of two variables}
\label{alg:smul}
\Procedure{SecureMUL$_\mathcal{AB}$}{$i_1, i_2, j$, level$_1$, level$_2$, opLevel}
    \State $\mathcal A$ finds $v_1$ with key $i_1$
    \State $\mathcal B$ finds $v_2$ with key $i_2$
    \State $\mathcal A$ encrypts $v_1$ and sends $\text{Enc}(v_1)$ to $B$
    \State $\mathcal B$ computes $\text{Enc}(v_1) \otimes v_2$ and stores the result with key $j$
\EndProcedure

\Procedure{SecureMUL$_\mathcal{IA}$}{$i_1, i_2, j$, level$_1$, level$_2$, opLevel}
    \State $\mathcal B$ finds $c_1$ with key $i_1$
    \State $\mathcal A$ finds $v_2$ with key $i_2$
    \State $\mathcal B$ samples $r \stackrel{\$}{\leftarrow} \mathbb Z_N$ and computes $c_1' = c_1 \oplus \text{Enc}(-r)$, then sends $c_1'$ to $A$
    \State $\mathcal A$ computes $c_3' = c_1'\otimes v_2$, encrypts $v_2$, and sends $c_3'$ and $c_2 = \text{Enc}(v_2)$ to $B$
    \State $\mathcal B$ computes $c_3' \oplus (c_2 \otimes r) \oplus \text{Enc}(0)$ and stores the result with key $j$
    \State // Adding $\text{Enc}(0)$ is to re-randomize the ciphertext
\EndProcedure

\Procedure{SecureMUL$_\mathcal{IB}$}{$i_1, i_2, j$, level$_1$, level$_2$, opLevel}
    \State $\mathcal B$ finds $c_1,v_2$ with key $i_1, i_2$
    \State $\mathcal B$ computes $c_1 \otimes v_2$ and stores the result with key $j$
\EndProcedure

\Procedure{SecureMUL$_\mathcal{II}$}{$i_1, i_2, j$, level$_1$, level$_2$, opLevel}
    \State $\mathcal B$ finds $c_1, c_2$ with key $i_1, i_2$
    \State $\mathcal B$ samples $r_1, r_2 \stackrel{\$}{\leftarrow} \mathbb Z_N$
    \State $\mathcal B$ computes $c_1' \gets c_1 \oplus \text{Enc}(r_1)$ and $c_2' = c_2 \oplus \text{Enc}(r_2)$
    \State $\mathcal B$ sends $c_1', c_2'$ to $A$
    \State $\mathcal A$ decrypts $v_1' \gets \text{Dec}(c_1')$
    \State $\mathcal A$ computes $c_3' \gets c_2'\otimes v_1' \oplus \text{Enc}(0)$, and sends it to $\mathcal B$.
    \State $\mathcal B$ computes $c_3'\oplus \text{Enc}(-r_1r_2)\oplus (c_1 \otimes (-r_2)) \oplus (c_2 \otimes (-r_1))$ and stores the result with key $j$
\EndProcedure
\end{algorithmic}
\end{algorithm}

\subsubsection{Putting All Together}
With the $\mathsf{SecureAdd}$ and $\mathsf{SecureMul}$ protocols, We are now able to securely evaluate any polynomial.
Given a polynomial, we can find the corresponding circuit.
Then we perform the local computation to obtain a compact form of the circuit and compute the levels of all intermediate variables.
Finally, we securely evaluate the compact circuit gate by gate using the $\mathsf{SecureAdd}$ and $\mathsf{SecureMul}$ protocols.

\begin{algorithm}[h!]
\begin{algorithmic}[1]
\footnotesize
\caption{Secure two-party polynomial evaluation}
\label{alg:poly-compute}
\Procedure{SecurePoly}{$G, W$}
    \State $(G, W) \gets \text{LocalCompute}(G, W)$
    \State Levels$\gets \text{ComputeLevels}(G, W)$
    \For{$i = 1$}{$|G|$}
        \If{$G[i]$ is party $\mathcal C\in \{\mathcal A,\mathcal B\}$'s input variable}
            \State $\mathcal C$ stores the actual value of $\mathsf{Encode}(G[i])$ with key $i$
        \Else{ // $G[i] \in \{\mathsf{Add}, \mathsf{Mul}\}$}
            \State Perform $\mathsf{SecureAdd}$/$\mathsf{SecureMul}$ protocol according to $G[i]$ and the types of $G[j], j\in W[i]$, with the arguments $W[i, 1]$, $W[i, 2]$, $i$, $\text{Levels}[W[i, 1]], \text{Levels}[W[i, 2]], \text{Levels}[i]$.
        \EndIf
    \EndFor
    \State // Reveal the result
    \State $\mathcal B$ finds $c$ with key $|G|$

    \If{$\mathcal A$ is supposed to know the result}
        \State $\mathcal B$ sends $c$ to $\mathcal A$, who decrypts it to get $v \gets \text{Dec}(c)$, and compute $v \gets \mathsf{Decode}(v / S^{\text{Levels}[|G|] - 1})$
    \ElsIf{$\mathcal B$ is supposed to know the result}
        \State $\mathcal B$ adds a random number $r \stackrel{\$}{\leftarrow} \mathbb Z_N$ and send $c' = c \oplus \text{Enc}(r)$ to $\mathcal A$.
        \State $\mathcal A$ decrypts it to get $v' = \text{Dec}(c')$ and sends it to $\mathcal B$.
        \State $\mathcal B$ recovers $f(\mathbf a, \mathbf b)^\text{(fixed)}$ and gets the final result using \cref{eq:decode}.
    \ElsIf{Both $\mathcal A$ and $\mathcal B$ are supposed to know the result}
        \State $\mathcal B$ sends $c$ to $\mathcal A$, who decrypts it to get $v \gets \text{Dec}(c)$, and compute $v \gets \mathsf{Decode}(v / S^{\text{Levels}[|G|] - 1})$
        \State $\mathcal A$ sends $v$ to $\mathcal B$.
    \EndIf
\EndProcedure
\end{algorithmic}
\end{algorithm}

\section{Data Packing}
\begin{figure}
    \centering
    \includegraphics[width=1\linewidth]{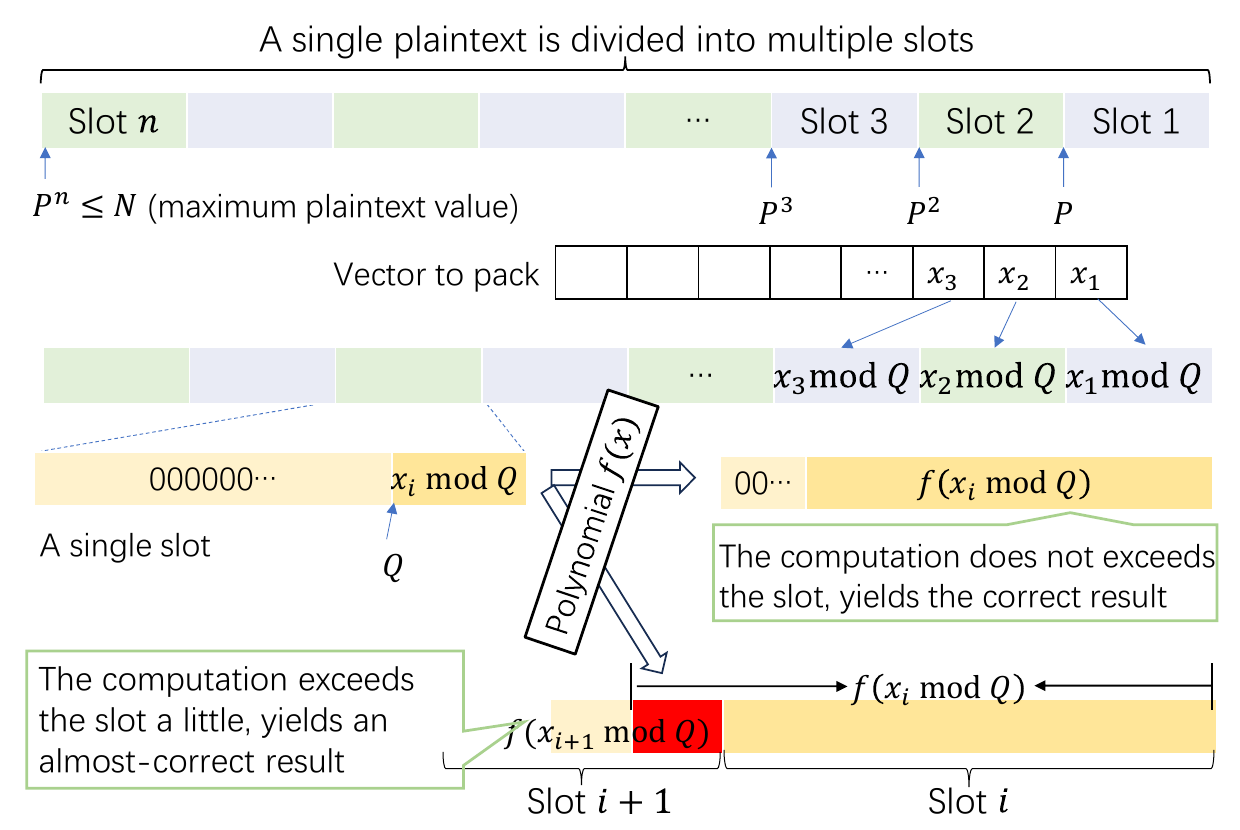}
    \caption{Illustration of our packing method.}
    \label{fig:data-packing}
\end{figure}
Data packing is used to accelerate cryptographic computations by putting multiple numbers into one big number that is still within the range of plaintext, which is widely used in homomorphic encryptions and multiparty computation~\cite{sharemind3_2012,aby2015,gazelle2018,chenhao2019multikey,erkin2012he_rec}.
In our case, we apply data packing to Paillier AHE which is operated on the integer ring.
Paillier cryptosystem usually uses a plaintext length of around 2048 bits to reach the security standard.
However, we do not need so many bits to encode just one number, as most operations in machine learning are performed on 32-bit float values.
If so, most bits in the plaintext are wasted as they will stay zero all through the computation.
Hence, it is necessary to implement data packing before cryptographic operations to improve efficiency.

The classical data packing method for Additive HE~\cite{point_inclusion_2007,bianchi2010encrypted_signal,wuweibin2023vpip} can be described as follows.
Suppose a cryptosystem accepts plaintext on the integer ring $\mathbb{Z}_N$, i.e., the plaintext is non-negative and smaller than $N$.
Then we can pack $n$ integers $x_1, x_2, ..., x_n$ to one plaintext as follows:
\begin{equation}
\label{eq:basic-pack}
    \mathsf{Pack}(x_1, ..., x_n) = x_1 + x_2\cdot P + x_3\cdot P^2 + ... + x_n \cdot P^{n-1};
\end{equation}
and the corresponding unpacking function:
\begin{equation}
\small
\label{eq:basic-unpack}
    \mathsf{Unpack}(X) = (\lfloor \dfrac{X}{P^0} \rfloor \bmod P, ..., \lfloor \dfrac{X}{P^{n-1}} \rfloor \bmod P).
\end{equation}

It is easy to verify that, as long as $0 \le x_i < P$ and $\sum_{i=1}^n (P-1)P^{i-1} = P^{n} - 1 < N$, the packed number will not exceed $N$, and unpacked numbers are exactly the original numbers.
In practice, $P$ is usually chosen as $2^L$ for efficiency, as the data format in the computer is binary.
Thus, if the bit-length of $N$ is $B+1$ ($2^{B} \le N < 2^{B + 1}$), we can pack $\lfloor B/L \rfloor$ integers into one plaintext.
For unpacking the plaintext, we just extract the $(i - 1)\cdot L$ to $i\cdot L - 1$ bits to form $x_i$.

\subsection{Optimal Data Packing}
Existing work on AHE only applies data packing to simple operations, e.g., element-wise addition and scalar-vector multiplication~\cite{point_inclusion_2007,erkin2012he_rec,shamir1979share,aby2015},
while little consideration on negative numbers, fixed-point computation, and more complex arithmetic circuits.
Hence, in this section, we design an optimal packing method for any fixed-point polynomial computation.
We display an overview of our packing method in \Cref{fig:data-packing}.

\subsubsection{Negative Values}
Without data packing, the arithmetic operations are performed on ring $\mathbb Z_N$, thus we can use $N - x$ to represent $x$ since $N$ is large~\cite{chaidi2021federated_mf,liuximeng2016pp_float,franz2011}.
However, this conversion cannot be applied to the data packing case as $\mathsf{Pack}(x_1, x_2, \cdots) \bmod N \ne \mathsf{Pack}(x_1 \bmod N, x_2 \bmod N, \cdots)$.
To mitigate this, we must select another modulus $Q < P$, and use $Q - x$ to represent $-x$.
Thus, we consider the computation to be performed on $\mathbb Z_Q$ instead of $\mathbb Z_N$.

\subsubsection{Choosing $P$ According to $Q$}
To produce the correct result using data packing, if we pack $x$ when computing an integer polynomial $f(x, \mathbf y) = f(x, y_1, \cdots, y_n)$ on $x_1$, we must make sure that
\begin{equation}
\label{eq:correct-pack}
\begin{split}
      \mathsf{Unpack}[f(\mathsf{Pack}(x_1,\cdots, x_n), \mathbf y)] 
      \\= [f(x_1, \mathbf y), \cdots, f(x_n, \mathbf y)].
\end{split}
\end{equation}
We now give the required conditions for this.
\begin{theorem}[Packed computation on integer]
\label{thm:int-cond}
Suppose that $x_1, \cdots, x_n$ are $n$ integers satisfying $P^{n} \le N$, $f(x_i, \mathbf y) < P$, and $x_i$'s degree is less than 2 in $f(x_i, \mathbf y)$, then \Cref{eq:correct-pack} holds.
\end{theorem}
\begin{proof}
Since $x_i$'s degree is less than 2, $f(x_i, \mathbf y)$ is a linear function of $x_i$.
Thus we have
\begin{equation}
\small
\begin{split}
    &f(\mathsf{Pack}(x_1,\cdots,x_n), \mathbf y) 
     = f\left(\sum_{i = 1}^n x_i P^{i-1}, \mathbf y\right) 
    \\ & = \sum_{i = 1}^n P^{i-1} f(x_i, \mathbf y) 
     \le \sum_{i=1}^n P^i < N
\end{split}
\end{equation}
By \Cref{eq:basic-unpack}, we can extract the $j$-th term as follows:
\begin{equation}
\small
\begin{split}
    & \left\{ \mathsf{Unpack}[f(\mathsf{Pack}(x_1, \cdots, x_n), \mathbf y)] \right\}_i \\
    &= \left\lfloor \sum_{i=1}^n P^{i-1}f(x_i, \mathbf y) / P^{j-1} \right\rfloor \bmod P \\
    &= \sum_{i = j}^n P^{i - j} f(x_i, \mathbf y) \bmod P = f(x_j, \mathbf y).
\end{split}
\end{equation}
The second equality holds since 
\\
$\sum_{i=1}^{j-1} P^{i - 1}f(x_i, \mathbf y) \le \sum_{i=1}^{j-1} P^{i-1}(P-1) < P^{j-1}$.

\end{proof}

\Cref{thm:int-cond} provides conditions for packed computation on integer polynomials, however, we are interested in the real-valued polynomials in real applications.
Hence, we have to turn a real-valued polynomial into a fixed-point one according to our computation method proposed in \Cref{sec:two-party-computation}.
Recall that we use a scale factor $S$ to convert a real variable into an integer variable, and each multiplication results in an extra multiplied $S$.
Thus, if we write a real-valued polynomial of degree-$d$ in a monomial form 
\begin{equation}
    f(\mathbf x) = \sum_{i=1}^d\sum_{j = 1}^{n_i} c_{i,j} M_{i,j}(x_1, \cdots, x_n),
\end{equation}
where $M_{i,j}$ is a degree-$i$ monomial with no coefficient.
The fixed-point version (ignoring the small rounding error) is
\begin{equation}
\small
\begin{split}
    f^\text{(fixed)}(\mathbf x) =& \sum_{i=1}^d\sum_{j = 1}^{n_i} (c_{i,j} \bmod Q)
    \\ &\cdot M_{i,j}(x_1S \bmod Q, \cdots, x_nS \bmod Q)\cdot S^{d - i}.
\end{split}
\end{equation}
Since $c_{i, j} \bmod Q > 0$ and $0 \le x_iS \bmod Q <Q$, we have $f^\text{(fixed)}(\mathbf x) \le f^\text{(fixed)}(Q, \cdots, Q)$.
In conclusion, to correctly compute polynomials with data packing, we have to choose 
\begin{equation}
\label{eq:exact-pack}
    P > f^\text{(fixed)}(Q, \cdots, Q).
\end{equation}
We denote this packing method as \textit{exact packing} since the packed computation result is exactly the same with the no-packing case.
\subsubsection{Optimizing $P$ According to Error-Bound}
While in \eqref{eq:correct-pack}, we assume the unpacked numbers are strictly equal to the actual ones.
This is not essential for correct computation, since the lower bits are negligible after being converted to float values.
Specifically, as we choose the scale factor to be $S$, we allow an error of $S^{l-1}$ for values of level $S^l$ so that the corresponding real-valued error is at most $1/S$.
Thus, \eqref{eq:correct-pack} becomes
\begin{equation}
\begin{split}
    \mathsf{Unpack}[f(\mathsf{Pack}(x_1, \cdots, x_n), \mathbf y)]_i \le f(x_i, \mathbf y) \cdot S^{l-1}.
\end{split}
\end{equation}
To do this, we only need a small modification of \Cref{thm:int-cond} by changing \eqref{eq:exact-pack} to 
\begin{equation}
    S^{l-1} \cdot P > f^\text{(fixed)}(Q, \cdots, Q)  \text{ and } N \ge S^{l-1}P^{n}.
\end{equation}
Through this, we can make $P$ smaller so that more values can be packed into one plaintext.
Notice that in this case, we require $P$ to be a multiple of $Q$.
This is easy to achieve since we choose both $P, Q$ as powers of 2.
We denote this packing method as \textit{approximate packing} as certain less significant bits in computation results are not correct.

\subsubsection{Choose $Q$ According to $S$}
Since we use $Q$ as the modulus for encoding, we have $\mathsf{Encode}(x;Q,S) = x$ when $x \in \left[\dfrac{Q}{2S}, \dfrac{Q}{2S} \right)$.
We notice that $f^\text{(fixed)}$ with degree $d$ has a scale factor of $S^d$, hence, to ensure correct decoding, we have to ensure that $f(\mathbf x) \in \left[\dfrac{Q}{2S^d}, \dfrac{Q}{2S^d} \right)$.
This leads to the choice of $Q > 2S^d \sup |f(\mathbf x)|$.

\subsubsection{Example}
To show the advantage of our data packing method, we show an example as follows.
Suppose $N > 2^{2048}$ and the function to compute is $\sum_{j=1}^{16} x_iy_j$, where $x_i, y_j < Q = 2^{55}, S = 2^{23}$
(this means the corresponding real-value computation result lies in $[-256, 256)$).
We want to pack multiple $x_i$'s into one plaintext.
In the conventional Paillier scheme, one plaintext only contains one value.
Using the exact packing \eqref{eq:exact-pack}, the mininum slot size is $16\cdot (2^{55})^2 = 2^{114}$, and using the error-bound-based packing, the minimum slot size is reduced to $2^{114 - 23} = 2^{91}$.
The numbers of slots using these two packing methods are $17$ and $22$ respectively.
Thus, we can achieve up to $\approx 22\times$ speed up in both computation and communication.
%


\subsection{Secure Repacking}
In many cases, we only apply data packing on a certain variable during a specific computation process, instead of packing the same variable throughout the computation.
In order to make $\mathcal A$ correctly unpack and repack the intermediate values, we can not simply sample the mask $r$ from $\mathbb Z_N$ and apply it to the packed value.
Instead, the mask $r$ is sampled for each packed element individually, and $x + r < P$ for each packed element $x$.
Thus, $\mathcal A$ can decrypt each masked element, and rearrange them in any desired order.
Specifically, we describe the SecureRepack protocol in \Cref{alg:srepack}.
Here, we use two matrices $I\in \mathbb Z^{m\times n}$ and $J\in \mathbb Z^{m' \times n'}$ to denote the previous packing layout and repacked layout.
Specifically, the elements in $I$ and $J$ are unique keys of the packed variables, and the corresponding values of each row are packed in one ciphertext.
The numbers of columns (i.e. $n, n'$) denote the numbers of elements in a packed ciphertext, 
and the numbers of rows (i.e. $m, m'$) denote the total numbers of ciphertexts.

\begin{algorithm}[h!]
\begin{algorithmic}[1]
\small
\caption{Secure repacking of intermediate variables}
\label{alg:srepack}
\Procedure{SecureRepack}{$I\in \mathbb Z^{m\times n}, J\in \mathbb Z^{m'\times n'}$}
    \State $\mathcal B$ creates an array $E_I\gets [\ ]$ to store masked packed ciphertexts
    \ForAll{row in $I$}
        \State $\mathcal B$ retrieves the corresponding encrypted packed value $C = \text{Enc}[\mathsf{Pack}(c_{i, 1}, \cdots, c_{i, n})]$
        \State $\mathcal B$ randomly picks $|C|$ random elements $r_{i, 1}, \cdots, r_{i, n} \stackrel{\$}{\gets} \mathbb Z_{N - X}$ 
        \State // $X$ is the upper bound of the values
        \State $\mathcal B$ computes $C' = C \oplus \text{Enc}[\mathsf{Pack}(r_{i,1}, \cdots, r_{i, n})]$
        \State $\mathcal B$ adds $C'$ to $E_I$
    \EndFor
    \State $\mathcal B$ sends $E_I$ to $\mathcal A$
    \State $\mathcal A$ creates an array $F_I\gets [\ ]$ to store individual masked elements
    \For{$i=1$ to $m$}
        \State $A$ decrypts $E_I[i]$ and gets $x_{i, 1} + r_{i, 1}, \cdots, x_{i, n} + r_{i, n}$
        \For{$j=1$ to $n$}
            \State $A$ adds tuple $(I_{i,j}, x_{i, j} + r_{i, j})$ to $F_I$
        \EndFor
    \EndFor
    \State $\mathcal A$ creates an array $E_J$ to store packed ciphertexts
    \For{$k = 1$ to $m'$}
        \State $\mathcal A$ creates an array $X\gets [\ ]$ to store packed elements
        \For{$l = 1$ to $n'$}
            \State $\mathcal A$ adds $E_I[i',j']$ to $X$, where $I_{i',j'} = J_{k,l}$
        \EndFor
        \State $\mathcal A$ adds $\text{Enc}[\mathsf{Pack}(X)]$ to $E_J$
    \EndFor
    \State $\mathcal A$ sends $E_J$ to $\mathcal B$.
\EndProcedure
\end{algorithmic}
\end{algorithm}

\section{Decentralized Social Matrix Factorization}

In this section, we describe the problem of decentralized social recommendation and put forward two solutions. The first one is based on bipartite computation, and the second, which is our principal approach known as PADER, involves computation in a natural order. Additionally, data packing is utilized in both solutions.
\subsection{Problem Definition}
First, we formally define the decentralized social recommendation problem as follows:
\begin{enumerate}
    \item Each user or seller (shop, restaurant, website, etc) is a node in the decentralized network.
    \item The sellers keep their own item embeddings private, and the users also keep their embeddings private.
    \item After a user interacts (e.g., viewing, clicking, purchasing, ordering a dish) with an item, the user has a rating of the item, e.g., like/dislike, or a score between $1 \sim 5$. We assume those interactions and ratings are also sensitive and being kept private by users.
    However, the seller also knows the item set the user interacted with because the item interacted/rated by the user must be provided by the seller.
    For example, the items are displayed on the website (so that the user can have different kinds of interactions) or sold in offline shops (so that the user can give a rating).
    Therefore, we apply weights to different interactions to model different interaction types.
    \item Users' social relations (e.g., friendships with other users) are kept secret from the sellers or any other users.
\end{enumerate}
In conclusion, our problem can be summarized as learning user/item embeddings $U$ and $V$ to approximate the rating matrix $R$, while keeping $U$, $V$, $R$, and the social relationships private to their owners, while the indices of non-zero entries in $R$ are shared among the corresponding user and seller.

\subsection{Secure SGD for Social Recommendation}
Matrix factorization can be efficiently solved via the SGD (Stochastic Gradient Descent) algorithm~\cite{zhangtong_sgd_2004}.
Since our recommendation system is decentralized, we consider each SGD step to be performed between a seller (denoted as $\mathcal A$ and we assume it has the secret key) and a user (denoted as $\mathcal B$).
Suppose that the user has interacted with $n$ items of the seller with the ratings $r_1, \cdots, r_n$ and weights $w_1, \cdots, w_n$.
The user's embedding is $\mathbf u$, and the items' embeddings are $\mathbf v_1, \cdots, \mathbf v_n$.
Also, we assume that the user has $m$ friends, whose embeddings are $\mathbf f_1, \cdots, \mathbf f_m$.
Thus, the gradient computation between one user and one recommender can be written as:
\begin{equation}
\small
\label{eq:mf-sgd}
\begin{cases}
\begin{split}
    \dfrac12 \nabla_{u_p} = & \sum_{i=1}^n\left[w_i\left(\sum_{q=1}^ku_qv_{i,q} - r_i\right)v_{i, p}\right] 
    \\ & + \lambda_S \dfrac1m \sum_{j=1}^m (f_{j,p} - u_p),
\end{split}
    \\[1em]\displaystyle
    \dfrac12 \nabla_{v_p} = \sum_{i=1}^n w_i \left(\sum_{q=1}^ku_qv_{i.q} - r_i\right)u_p.
\end{cases}
\end{equation}
For convenience, we eliminate the bias term and the regularization term in \eqref{eq:smf-loss}.
The bias term can be represented as a fixed constant 1 in the user/item embeddings, while the regularization term can be computed locally by shrinking the embeddings each round.

We can see that \eqref{eq:mf-sgd} is a polynomial of the user and seller's inputs, along with the friends' embeddings which can directly become an intermediate value.
To do this, a friend of the user just encrypts his user embedding using the seller's public key and sends the encrypted value to the current user.

After the SGD computation of one iteration is finished, the gradient on user embedding $\partial L/\partial u_p$ is revealed to the user, but the gradients of item embeddings $\partial L/\partial v_p$ shall not be directly revealed to the recommender since it is proportional to $\mathbf u$ which leaks the user privacy.
To tackle this problem, we can simply perform an aggregation on $\partial L/\partial v_p$ by summing the gradient computed by different users together before revealing it to the seller.

\subsection{Bipartite Computation}
To perform bipartite computation described in \Cref{alg:bipartite}, we rewrite Equation \eqref{eq:mf-sgd} as follows:
\begin{equation}
\small
\label{eq:bipartite-sgd}
\begin{cases}
\begin{split}
    \dfrac12 \nabla_{u_p} =
    & \sum_{i=1}^n \sum_{q=1}^k \underbrace{w_iu_q}_{\mathcal B}\underbrace{v_{i,q}v_{i,p}}_{\mathcal A}
    - \sum_{i=1}^n \underbrace{w_ir_{i}}_{\mathcal B} \underbrace{v_{i,p}}_{\mathcal A}\\
    +& \sum_{j=1}^m \underbrace{(\lambda_S/m)}_{\mathcal B}\underbrace{f_{j,p}}_{\mathcal I}
    - \sum_{j=1}^m \underbrace{(\lambda_S/m) u_p}_{\mathcal B}, 
\end{split}
    \\ \displaystyle       
    \dfrac12 \nabla_{v_p} = 
    \sum_{i=1}^n \sum_{q=1}^k \underbrace{w_iu_qu_p}_{\mathcal B}\underbrace{v_{i,q}}_{\mathcal A} 
    - \underbrace{ \sum_{i=1}^n w_ir_{i}u_p}_{\mathcal B}.
\end{cases}
\end{equation}
We can see that each term in the equation is either computed on one party or a product of two parties' locally computable values.
The detailed steps of this bipartite computation are presented in Algorithm 8.

\begin{algorithm}[H]
\small
\caption{Bipartite Computation of Secure SGD}
\begin{algorithmic}[1]

\Procedure{SecureSGD}{Seller $\mathcal A$, User $\mathcal B$} \Comment{Bipartite}
\State $\mathcal B$ initializes embedding \(\mathbf u_p\), ratings \(r_i\) and interaction weights \(w_i\) (\(i = 1\) to \(n\))
\State $\mathcal A$ initializes embedding \(\mathbf {v_i}\) (\(i = 1\) to \(n\))


    \State $\mathcal A$ computes and sends \(\text{Enc}(\mathsf{Pack}_{p}[v_{i,q} v_{i,p}])\), \(\text{Enc}(\mathsf{Pack}_{p}[v_{i,p}])\), \(\text{Enc}(v_{i,q})\) (\(i = 1\) to \(n\), \(q = 1\) to \(k\)) to $\mathcal B$.
    \State Friend \(j\) of $\mathcal B$ computes and sends \(\text{Enc}(\mathsf{Pack}_{p}[f_{j,p}])\) to $\mathcal B$ (\(j = 1\) to \(m\))

    \State $\mathcal B$ computes \(\text{Enc}(\mathsf{Pack}_{p}[\nabla_{u_p}])\) and \(\text{Enc}(\mathsf{Pack}_{p}[\nabla_{v_p}])\) using \Cref{eq:bipartite-sgd} with \(\mathsf{SecureAdd}\) and \(\mathsf{SecureMul}\).
    \State $\mathcal B$ masks \(\text{Enc}(\mathsf{Pack}_{p}[\nabla_{u_p}])\) as \(\text{Enc}(\mathsf{Pack}_{p}[\nabla_{u_p}^\text{mask}])\), sends \(\text{Enc}(\mathsf{Pack}_{p}[\nabla_{u_p}^\text{mask}])\) and   \(\text{Enc}(\mathsf{Pack}_{p}[\nabla_{v_p}])\) to $\mathcal A$.
    
    \State $\mathcal A$ decrypts and unpacks \(\text{Enc}(\mathsf{Pack}_{p}[\nabla_{v_p}])\) to get \(\nabla_{v_p}\) and updates \(v_p\).
    
    \State $\mathcal A$ decrypts and unpacks \(\text{Enc}(\mathsf{Pack}_{p}[\nabla_{u_p}^\text{mask}])\) to get \(\nabla_{u_p}^\text{mask}\) and sends it to $\mathcal B$.

    \State $\mathcal B$ unmasks \(\nabla_{u_p}^\text{mask}\) to get \(\nabla_{u_p}\) and updates \(u_p\).

\EndProcedure

\end{algorithmic}
\end{algorithm}

\subsubsection{Data Packing}
Since the computation is performed on all $k$ dimensions of the embedding, we apply data packing to the embedding dimension, i.e., the $p$ subscript in $v_{i,q}v_{i,p}, v_{i,p}$, and $ f_{j,p}$.
The packed ciphertexts are represented by $\text{Enc}(\mathsf{Pack}_p[v_{i,p}]), \text{Enc}(\mathsf{Pack}_p[v_{i,q}v_{i,p}])$, and $\text{Enc}(\mathsf{Pack}_p[f_{j,p}])$.
Here, we use $\mathsf{Pack}_p$ to represent the packing along the subscript $p$.
As for the choice of $P$ given $Q$, we have: $P > (nk + n + m)Q^2 + mQS$ using \Cref{thm:int-cond}.

\subsubsection{Communication}
The bipartite computation only requires a constant number of communication rounds.
In this case, three transmissions between the user and the seller are needed.

\begin{enumerate}
    \item \textbf{(Algorithm 8, Line 4):} 
    The seller sends all needed ciphertexts to the user.
    There are a total of $nk' + nkk' + nk$ ciphertexts.
    %
    \item \textbf{(Algorithm 8, Line 7):} 
    The user sends the ciphertext of item gradients and masked user gradient to the seller. There are a total of $nk' + k'$ ciphertexts.
    \item \textbf{(Algorithm 8, Line 9):} 
    The seller sends the plaintext of the masked user gradient to the user.
    There are a total of $k'$ plaintexts.
\end{enumerate}
Here $k'$ is the number of ciphertexts when $k$ plaintexts are packed.
The total communication size between the user and recommender is $nkk' + 2nk' +  nk+ 2k'$ when data packing is applied, and is $nk^2 + 2nk + 2k$ otherwise.

\subsection{Natural Order Computation}
A natural way to compute Equation \eqref{eq:mf-sgd} is to first calculate the error term $e_i$ as follow:
\begin{equation}
\label{eq:e_i}
e_i = w_i \sum_{q=1}^k \left(u_qv_{i,q} - r_i\right) = 
\sum_{q=1}^k (\underbrace{w_iu_q}_\mathcal B \underbrace{v_{i,q}}_\mathcal A - \underbrace{r_i}_\mathcal B)
\end{equation}
Next, we compute the gradients as follows:
\begin{equation}
\label{eq:natural-sgd}
\begin{cases}
\begin{split}
    \dfrac12 \nabla_{u_p} = \sum_{i=1}^n \underbrace{e_i}_\mathcal I \underbrace{v_{i, p}}_\mathcal A  + \lambda_S \dfrac1m \sum_{j=1}^m (\underbrace{f_{j,p}}_\mathcal I - \underbrace{u_p}_\mathcal B), 
\end{split}
    \\[2em] \displaystyle
    \dfrac12 \nabla_{u_p} = \sum_{i=1}^n \underbrace{e_i}_\mathcal I \underbrace{u_p}_\mathcal B.

\end{cases}
\end{equation}
While it is not a bipartite form, our proposed protocol enables this computation.
%
%
The specific process of this natural order computation is demonstrated in Algorithm 9.

\begin{algorithm}[H]
\small
\caption{Natural Order Computation of Secure SGD}
\begin{algorithmic}[1]
\Procedure{SecureSGD}{Seller $\mathcal A$, User $\mathcal B$} \Comment{Natural order}
\State $\mathcal B$ initializes embedding \(\mathbf u_p\), ratings \(r_i\) and interaction weights \(w_i\) (\(i = 1\) to \(n\))
\State $\mathcal A$ initializes embedding \(\mathbf {v_i}\) (\(i = 1\) to \(n\))

    \State $\mathcal A$ sends \(\text{Enc}(\mathsf{Pack}_{q}[v_{i,q}])\) (\(i = 1\) to \(n\)) to $\mathcal B$.
    \State $\mathcal B$ computes \(\text{Enc}(\mathsf{Pack}_{q}[e_i])\) using \Cref{eq:e_i} with \(\mathsf{SecureAdd}\) and \(\mathsf{SecureMul}\).
    \State $\mathcal B$ masks \(\text{Enc}(\mathsf{Pack}_{q}[e_i])\) as \(\text{Enc}(\mathsf{Pack}_{q}[e_i^\text{mask}])\) and sends it to $\mathcal A$.
    \State $\mathcal A$ computes \(\sum_{i=1}^n {e_i^\text{mask}} {v_{i, p}}\) with \(\mathsf{SecureAdd}\) and \(\mathsf{SecureMul}\) and repacks \(\text{Enc}(\mathsf{Pack}_{q}[e_i^\text{mask}])\) to get \(\text{Enc}(\mathsf{Pack}_{p}[e_i^\text{mask}])\) with \(\mathsf{SecureRepack}\)
    \State $\mathcal A$ sends \(\sum_{i=1}^n {e_i^\text{mask}}{v_{i, p}}\), \(\text{Enc}(\mathsf{Pack}_{p}[e_i^\text{mask}])\) and \(\text{Enc}(\mathsf{Pack}_{p}[v_{i,p}])\) to $\mathcal B$.
    \State $\mathcal B$ unmasks, get \(\sum_{i=1}^n {e_i}{v_{i, p}}\) and \(\text{Enc}(\mathsf{Pack}_{p}[e_i])\).
    \State $\mathcal B$ computes \(\text{Enc}(\mathsf{Pack}_{p}[\nabla_{u_p}])\) and \(\text{Enc}(\mathsf{Pack}_{p}[\nabla_{v_p}])\) using \Cref{eq:natural-sgd} with \(\mathsf{SecureAdd}\) and \(\mathsf{SecureMul}\).
    \State $\mathcal B$ masks \(\text{Enc}(\mathsf{Pack}_{p}[\nabla_{u_p}])\) as \(\text{Enc}(\mathsf{Pack}_{p}[\nabla_{u_p}^\text{mask}])\), sends \(\text{Enc}(\mathsf{Pack}_{p}[\nabla_{u_p}^\text{mask}])\) and   \(\text{Enc}(\mathsf{Pack}_{p}[\nabla_{v_p}])\) to $\mathcal A$.
    
    \State $\mathcal A$ decrypts and unpacks \(\text{Enc}(\mathsf{Pack}_{p}[\nabla_{v_p}])\) to get \(\nabla_{v_p}\) and updates \(v_p\).
    
    \State $\mathcal A$ decrypts and unpacks \(\text{Enc}(\mathsf{Pack}_{p}[\nabla_{u_p}^\text{mask}])\) to get \(\nabla_{u_p}^\text{mask}\) and sends it to $\mathcal B$.

    \State $\mathcal B$ unmasks \(\nabla_{u_p}^\text{mask}\) to get \(\nabla_{u_p}\) and updates \(u_p\).
\EndProcedure
\end{algorithmic}
\end{algorithm}

\subsubsection{Data Packing}
In the first step to compute $e_i$, we pack different items together, i.e., the packing is applied on subscript $q$.
In the second step, we pack along the embedding dimension as in the bipartite case, i.e., the packing is applied on subscript $p$.
As for the choice of $P$ given $Q$, we have $P > nk(3Q^3 + 2Q^2S) + mQ^3$.

\subsubsection{Communication}
In the natural order computation, five transmissions
between the user and the seller are needed.
\begin{enumerate}
    \item \textbf{(Algorithm 9, Line 4):} The seller first sends $n'k$ ciphertexts for encrypted item embedding packed along the the item dimension for the computation of $\sum w_iu_qv_{i,q}$.
    \item \textbf{(Algorithm 9, Line 6):}
    The user sends back $n'$ ciphertexts for masked ciphertext of $e_i$, where $n'$ is the number of ciphertexts when packing $n$ ciphertexts together.
    \item \textbf{(Algorithm 9, Line 8):} 
    The seller send back the masked $\sum e_iv_{i,p}$, which have a size of $k'$. Considering that the second step requires a different packing dimension with the first step, the seller has to repack the masked ciphertext of $e_i$, resulting in $n$ transferred ciphertexts. Also, The seller need to send $nk'$ ciphertexts for encrypted item embedding packed along the the embedding dimension for the computation of $e_iv_{i, p}$. There are a total of $k' + n + nk'$ ciphertexts.
    \item \textbf{(Algorithm 9, Line 11):} 
    The user sends the ciphertext of item gradients and masked user gradient to the seller. There are a total of $nk' + k'$ ciphertexts.
    \item \textbf{(Algorithm 8, Line 13):} 
    The seller sends the plaintext of the masked user gradient to the user. There are a total of $k'$ plaintexts.
\end{enumerate}
Overall, when using data packing, the total communication size is $n'k + n' + n + 3k' + 2nk'$, and $3nk + n + 3k$ without data packing.

\section{Security Analysis}
In this section, we conduct a comprehensive security analysis of the PADER system. The primary participants in the PADER system are users and sellers. Users are concerned with safeguarding their embeddings, rating data, and social relationships, while sellers aim to protect their item information. We begin by defining the threat model and then proceed to prove the security of the proposed protocol under the semi-honest adversary model, ultimately demonstrating that the PADER system effectively protects the privacy of its participants.
\subsection{Threat Model}
This paper focuses on the threat model of semi-honest adversaries. In this model, semi-honest adversaries will strictly adhere to the protocol for operations but may attempt to deduce sensitive information about other parties from the data obtained during protocol execution. For example, a user might try to extract other users' ratings of goods from the encrypted data provided by the seller, while the seller might attempt to obtain the user's social relationships and other private data through interactions with the user.

Under this threat model, our protocol aims to protect the following: (i) The embeddings, rating data, and social relationships of users are not known to the seller or other users; (ii) The item embeddings of the seller are not known to the user; (iii) All intermediate computation results remain confidential to all participants during the computation process, and only the final computation result will be disclosed to the authorized participants.

\subsection{Security of arithmetic circuits}

To prove that our protocol is secure in the presence of static semi-honest adversaries, we use the simulation-based method.
The simulation-based security proof can be described as follows:
For any corrupted party $\mathcal P$ (either $\mathcal A$ or $\mathcal B$ in our case), we can use a simulator $\mathcal S$ to simulate the view (defined as the messages received by $P$ during the execution of the protocol) of $\mathcal P$ based on its own input and the designated output, which the adversary can not distinguish with the view obtained by a real execution of the protocol.

\begin{theorem}
\label{thm:security}
    Our proposed protocol (\Cref{alg:poly-compute}) securely evaluates the arithmetic circuit in the presence of static semi-honest adversaries.
\end{theorem}
We prove the theorem in two steps, i.e., simulating $\mathcal A$'s view and $\mathcal B$'s view respectively.
First, we create a simulator $\mathcal S_A$ to simulate $\mathcal A$'s view.
The view of  $\mathcal A$ is 
\begin{equation}
\begin{split}
    \mathsf{view}^\text{real}_\mathcal A(x, y) = (&\text{Enc}(r_1), \cdots, \text{Enc}(r_{n_1}), 
    \\
    & \text{Enc}(s_1), \cdots, \text{Enc}(s_{n_2})),
\end{split}
\end{equation}
where we use $r_1, \cdots, r_{n_1}$ to denote the ciphertexts of masked elements in $\mathbb Z_N$ used in the $\mathsf{SecureMul}$ protocols, and $s_1, \cdots, s_{n_2}$ are the masked elements used in the $\mathsf{SecureRepack}$.
We use $x, y$ to denote the input for $\mathcal A$ and $\mathcal B$ respectively.

We notice that $r_1, \cdots, r_{n-1}$ are ciphertexts of independently uniformly random elements on $\mathbb Z_N$, which can be directly simulated.
On the other hand, the masked values used in repacking are not uniformly random due to the mask is not uniformly sampled from $\mathbb Z_N$ but from $\mathbb Z_M$ such that $P > M \gg X$, where $X$ is the predetermined upper bound of the values to repack.
In this case, for some value $z$, the masked value is drawn from $U[Z, Z+M)$.
However, although we do not know $x$, we can still sample uniformly from $U[0, M)$ to simulate the masked value.
It is easy to see, that those two distributions are identical except for an exponentially small (in terms of the security parameter) portion of $2X/M$.

Now we show that such simulation can generate the view that is indistinguishable from any probabilistic polynomial time algorithm $D$, i.e., for any such $D$. we have
\begin{equation}
\label{eq:indistinguishable-view}
\footnotesize
\begin{split}
        & |\text{Pr}[D(\mathsf{view}_\mathcal A^\text{real}) = 1] - \text{Pr}[D(\mathsf{view}_\mathcal A^\text{sim}) = 1]| 
       \le 1/\mu(n),
\end{split}
\end{equation}
where $\mu(n)$ is some negligible function, i.e., for any polynomial $p(n)$, $\mu(n) < 1/p(n)$ for sufficiently large $n$, and here $n$ is the security parameter.
Notice that, we can write one term in \eqref{eq:indistinguishable-view} by the Bayes' law as
\begin{equation}
    \text{Pr}[D(\mathsf{view}) = 1] = \sum_{V} P(\mathsf{view} = V)I_{D(V) = 1}.
\end{equation}
Hence, we only have to consider $\Delta V$ which is the set of $V$ such that $P[\mathsf{view}_\mathcal{A}^\text{real}(x,y) = V] \ne P[\mathsf{view}_\mathcal{A}^\text{sim}(x,y) = V].$
We can then find an upper bound of \eqref{eq:indistinguishable-view} to be 
\begin{equation}
    \sum_{V\in \Delta V} [\text{Pr}(\mathsf{view}_\mathcal{A}^\text{real} = V) + \text{Pr}(\mathsf{view}_\mathcal{A}^\text{sim} = V)].
\end{equation}
In our case, $\Delta V$ corresponds to the case such that $\bigcap_{i=1}^{n_2}\{V: r'_i \in [0, z_i) \cup [M, z_i + M)\}$, where $0 \le z_i < Z$ is the $i$-th value to repack.
Thus, we have: 
\begin{equation}
\small
\begin{split}
&\text{Pr}(\mathsf{view}_\mathcal{A}^\text{real} = \Delta V) = 1 - \prod_{i=1}^{n_2} (1 - z_i/M)
\\
&=1 - (1 - Z/M)^{n_2} \le n_2 Z/M,
\end{split}
\end{equation}
and similarly,
\begin{equation}
\small
\begin{split}
&\text{Pr}(\mathsf{view}_\mathcal{A}^\text{sim} = \Delta V) = 1 - \prod_{i=1}^{n_2} (1 - z_i/M) \le n_2 Z/M.
\end{split}
\end{equation}
Combine these together, we have:
\begin{equation}
\small
\begin{split}
    & |\text{Pr}[D(\mathsf{view}_\mathcal A^\text{real}) = 1] - \text{Pr}[D(\mathsf{view}_\mathcal A^\text{sim})= 1]| 
    \\
    & \le 2n_2 Z/M = 2n_2 Z/2^m,
\end{split}
\end{equation}
where $m = \log_2 M$ is proportional to the security parameter.
Hence, we can simulate $\mathcal A$'s view except for negligible probability.

Now consider the simulation of $\mathcal B$'s view, which can be written as
\begin{equation}
    \mathsf{view}_\mathcal{B} = (\text{Enc}(t_1), \cdots, \text{Enc}(t_{n_3})).
\end{equation}
The first step is to show that any ciphertext $\text{Enc}(t_i)$ can actually be viewed as a fresh new encryption of $t_i$.

\begin{proposition}
Suppose $c_0 = \text{Enc}(v_0; r_0)$ is a given ciphertext, and $c_1 = \text{Enc}(v_1, r_1)$ where $r_1$ is a random number uniformly drawn from $\mathbb Z^*_n$, then $c_0 \otimes c_1 = c_0 c_1 \bmod n^2$ has the same distribution with $\text{Enc}(v_0 + v_1, r)$ where $r$ is uniformly drawn from $\mathbb Z^*_n$.
\end{proposition}
\begin{proof}
By the definition of Paillier cryptosystem, we have
\begin{equation}
\small
\begin{split}
    c_0c_1 & = g^{v_0}r_0^n \bmod n^2 \cdot g^{v_1}r_1^n \bmod n^2
    \\
    &= g^{v_0 + v_1} (r_0r_1)^n \bmod n^2
    \\
    &= g^{v_0 + v_1 \bmod n} (g^{n \lfloor \frac{v_0 + v_1}{n} \rfloor}r_0r_1 \bmod n)^n \bmod n^2,
\end{split}
\end{equation}
which is an encryption of $v_0 + v_1 \bmod n$.
Notice that the random factor is $g^{n \lfloor \frac{v_0 + v_1}{n} \rfloor}r_0r_1 \bmod n$, which is still coprime with $n$.
Consider the multiplicative group $G = (\mathbb Z_n^*, \times)$, since the coset of the complete group is the group itself, we have:
\begin{equation}
    \left( g^{n \lfloor \frac{v_0 + v_1}{n} \rfloor}r_0 \bmod n \right) \times G = G.
\end{equation}

Thus, the product $c_0c_1$ is still a fresh new encryption since its random factor is uniformly distributed.
This completes our proof.
\end{proof}

Now we construct the simulated view for $\mathcal B$ as follows:
\begin{equation}
    \mathsf{view}_{\mathcal B}^\text{sim} = (\text{Enc}(t_1'), \cdots, \text{Enc}(t'_{n_3}),
\end{equation}
where $t_1, \cdots, t_{n_3}$ are random numbers in $\mathbb Z_n$.
Then distinguishing between the simulated view and the real view is equivalent to distinguishing two (randomized) Paillier ciphertexts, which is considered impossible for any probabilistic polynomial adversary due to the IND-CPA security of Paillier.

In summary, we can use ciphertexts of random numbers to simulate both $\mathcal A$ and $\mathcal B$'s view, leading to the security under the semi-honest adversary.
This completes the proof for \Cref{thm:security}

\subsection{Security of PADER}
Based on the security of arithmetic circuits as established in Theorem 2, PADER's security against static semi-honest adversaries can be confirmed. In PADER, interactions between users and recommenders are modeled as polynomial computations involving sensitive data such as user embeddings, ratings, social relationships, and item embeddings.

PADER's security rests on two key elements: the Paillier cryptosystem's IND-CPA security for encrypting data, ensuring adversaries cannot derive useful information from encrypted data; and our secure computation protocols that keep all intermediate computational results confidential until the final result is revealed to authorized parties.

By combining the robust encryption properties of the Paillier cryptosystem with our secure computation protocols, PADER effectively protects user and recommender data privacy, safeguarding sensitive information from static semi-honest adversaries.

\section{Experiments}
To demonstrate the efficiency and performance of PADER, we conduct experiments on the synthetic data and two real-world 
 social recommendation datasets, i.e., Epinions~\cite{hamedani2021trust_rec_epinions} and Douban~\cite{mahao2011social_regularization}.
To evaluate the efficiency of PADER, We report the results on the computation and communication overheads during training and inference, in comparison with the bipartite computation, with or without data packing, and the state-of-the-art fully homomorphic cryptosystem CKKS~\cite{ckks2017}.
Moreover, the validation loss during training is also measured to show the correctness of PADER.

\begin{figure*}[ht]
    \centering
    \includegraphics[width=1\textwidth]{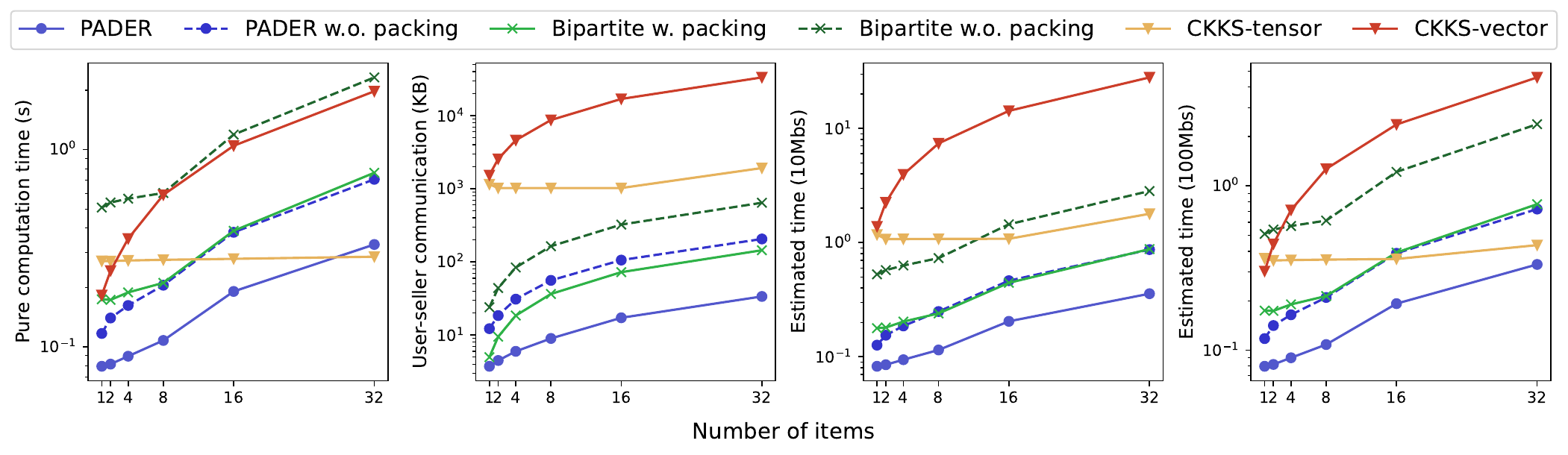}
    \caption{Results on different numbers of rated items, where the embedding dimension is 8 and the number of social connections is 10.
    The y-axis is in log-scale.}
    \label{fig:items}
\end{figure*}

\begin{figure*}[ht]
    \centering
    \includegraphics[width=1\textwidth]{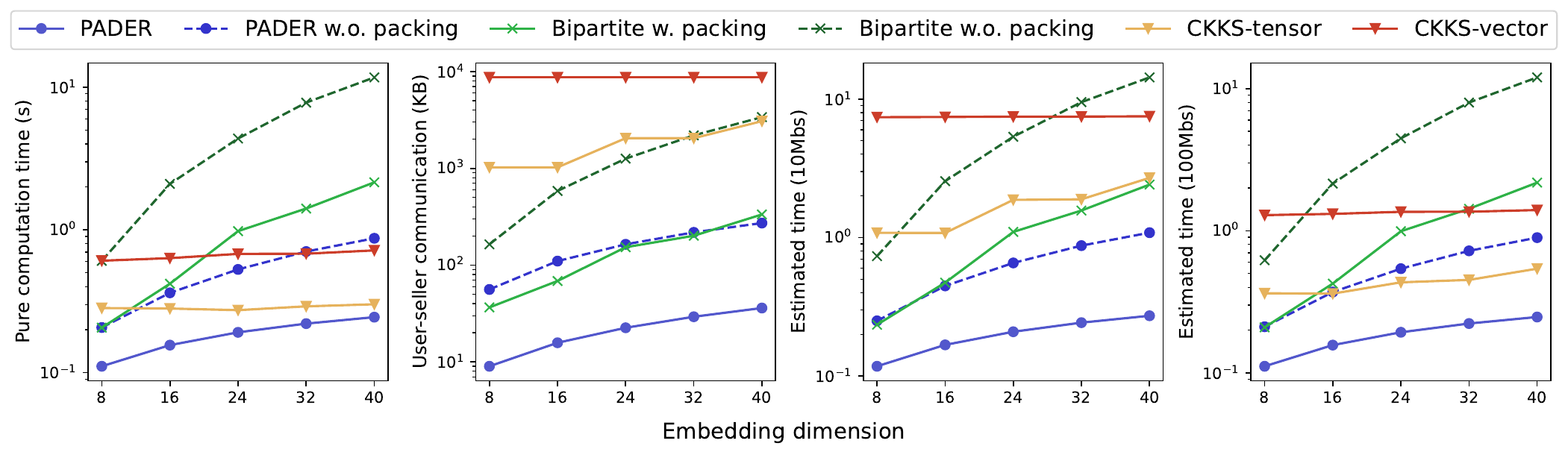}
    \caption{Results on different embedding dimensions, where the number of rated items is 8 and the number of social connections is 10.
    The y-axis is in log-scale.}
    \label{fig:dimensions}
\end{figure*}
\subsection{Experiment Setting}
For the implementation of the Paillier cryptosystem, we use Intel's Pallier Library~\cite{intel_paillier_lib}, which supports the DJN implementation for fast encryption/decryption~\cite{djn2010}.
The codes are compiled with OpenMP enabled, and the number of threads is set to 8 as we empirically found it has relatively good performance.
For all methods, we choose the precision to be 23 bits as suggested in \cite{courbariaux2014low_precision}, i.e., $S = 2^{23}$. The key length for Paillier is chosen to be 2048 to meet conventional security levels (112-bit security).
Based on this, we apply optimal data packing to both methods: for PADER (natural order computation), we choose $Q = 2^{80}$, and $P = 2^{256}$; for bipartite computation, we choose $Q = 2^{56}$ and $P = 2^{128}$.
Notice we choose $N$ to be the power of $2^{32}$ for the convenience of unpacking.
By this and given the plaintext length is 2048-bit, we can pack 8 plaintexts into one in PADER, and 16 in bipartite computation.

For the implementation of the CKKS cryptosystem, we adopt IBM's HElayers library~\cite{helayers}, which supports the Tile Tensor structure for high-efficiency data packing. The CKKS scheme is configured with a polynomial modulus degree of 8192, a scaling factor of $2^{30}$, and a ciphertext coefficient modulus of $2^{238}$. Consistent with the Paillier implementation, we also set the number of threads for HElayers to 8.
While the HElayers library supports both vector-based computation (which supports dot-product and vector-scalar multiplication) and tensor-based computation (which supports most matrix operations), we compare both in our experiments.

All experiments were conducted on a server featuring two Intel Xeon Silver 4116 CPUs (24 cores) @ 2.10 GHz.

\subsection{Training Benchmarks}
We conduct benchmarks on synthetic data to study the computation and communication cost in one training iteration regarding to number of items and embedding dimension.
The following results are reported:
\begin{itemize}
    \item Pure computation time, in which we ignore any communication costs.
    \item Total communication size between the user and seller.
    \item Estimated time consumption under low network bandwidth (10Mbps).
    \item Estimated time consumption under high network bandwidth (100Mbps).
\end{itemize}
The former two results are directly measured in experiments, as we simulate the network communication in a standalone server.
The latter two results are calculated as `pure time consumption + communication size/network bandwidth.

\begin{figure*}[ht]
    \includegraphics[width=\textwidth]{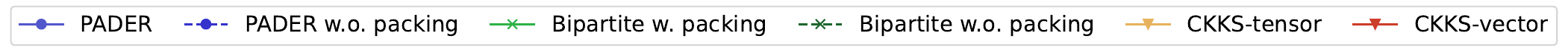}
    \centering
    \begin{subfigure}{0.48\textwidth}
        \includegraphics[width=\textwidth]{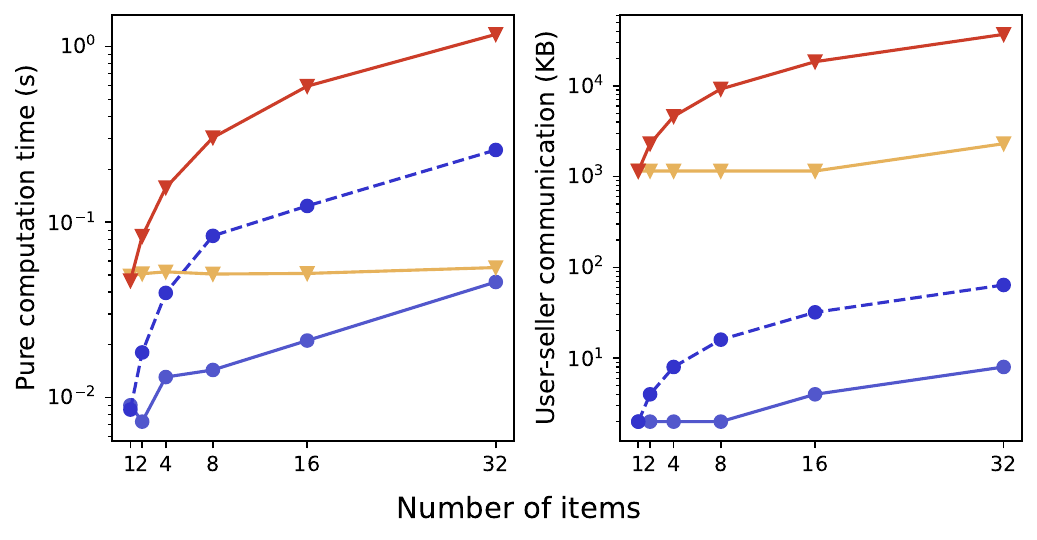}
        \caption{Inference cost vs. the number of items.}
        \label{fig:infer-item}
    \end{subfigure}
    \begin{subfigure}{0.48\textwidth}
        \includegraphics[width=\textwidth]{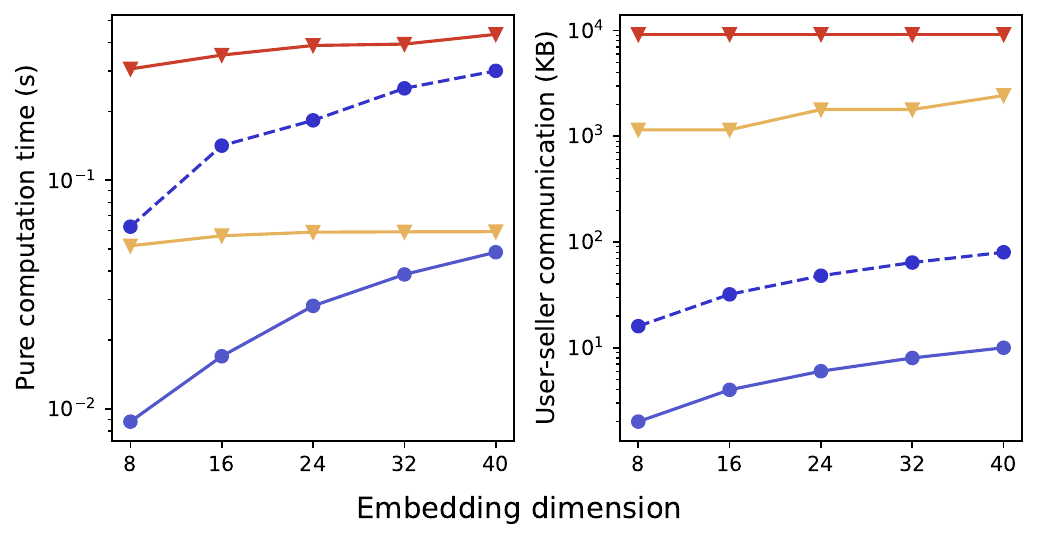}
        \caption{Inference cost vs. embedding dimension.}
        \label{fig:infer-dim}
    \end{subfigure}
    \label{fig:inference}
    \caption{Results on inference tasks.}
\end{figure*}


\begin{figure*}[ht]
    \includegraphics[width=\textwidth]{figs/legend.pdf}
    \centering
    \begin{subfigure}{0.48\textwidth}
        \includegraphics[width=\textwidth]{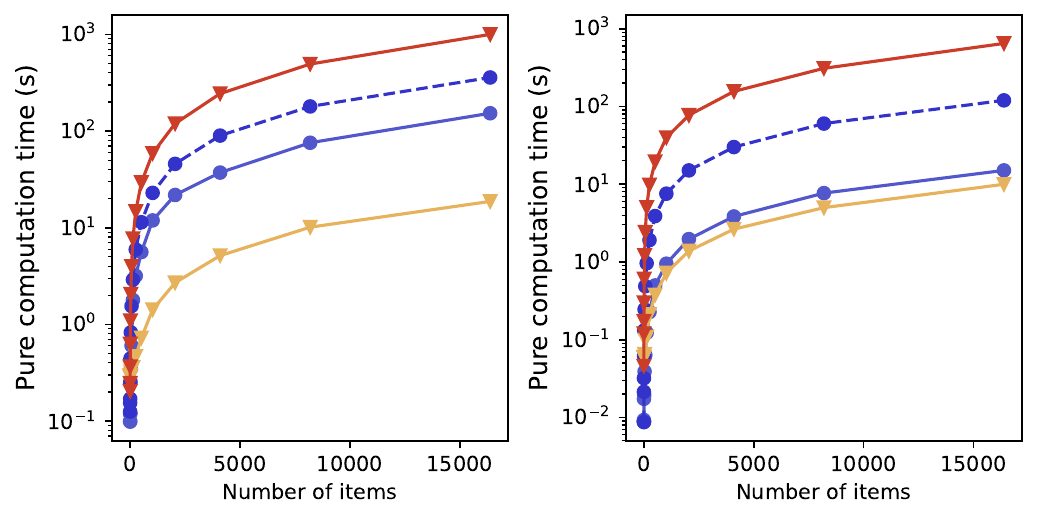}
        \caption{Training and inference pure computation time}
        \label{fig:limitation-time}
    \end{subfigure}
    \begin{subfigure}{0.48\textwidth}
        \includegraphics[width=\textwidth]{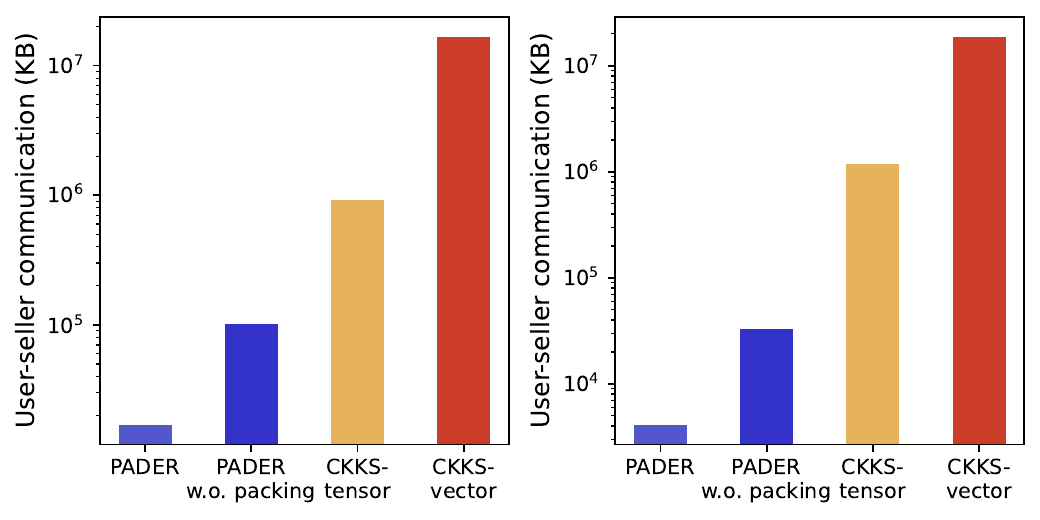}
        \caption{Training and inference communication volume}
        \label{fig:limitation-comm}
    \end{subfigure}
    \caption{Results on a large number of items. (a)~depicts the variation in pure computation time for training and inference as the item count scales from $2^1$ to $2^{14}$. (b)~compares the communication volume for training and inference at an item count of $2^{14}$.}
    \label{fig:limitation}
\end{figure*}

\subsubsection{Different Numbers of Rated Items}
The results on different numbers of rated items are reported in \Cref{fig:items},
In the experiment, we fix the embedding dimension to be 8 and the number of social connections to be 10, and vary the number of rated items from $1$ to $2^{5}$.

%
The results demonstrate that, for a small number of items, PADER achieves lower pure computation time compared to CKKS methods. However, when the number of items increases beyond $2^{5}$, the pure computation time of CKKS-tensor becomes superior to that of PADER. Notably, PADER's communication size is over two orders of magnitude smaller than that of CKKS methods. (In \Cref{fig:limitation}, we further examine scenarios with an exceptionally large number of items, which highlights a limitation of PADER.)

Also, PADER is more efficient than the bipartite computation based on Paillier AHE, and data packing also significantly reduces the computation and communication cost.
Hence, under lower bandwidth, PADER has a greater advantage over other methods.

\subsubsection{Different Embedding Dimensions}
The results on different embedding dimensions are reported in \Cref{fig:dimensions}.
In the experiment, we fix the number of rated items to be 8 and the number of social connections to be 10, and vary the embedding dimension in $\{8, 16, 24, 32, 40\}$.

Results show that our proposed PADER has the lowest costs like in the previous experiment.
%
It is noteworthy that with the CKKS methods, the costs remain nearly constant as the embedding dimension rises. The reason is that, under the common  setting, 4096 elements can be packed into a single ciphertext.
However, the results show that with the embedding dimension of up to 40, PADER still maintains significantly lower costs.
Especially in the communication cost, PADER constantly has at least a two-magnitude advantage compared with CKKS methods.
%




\subsection{Inference Benchmarks}
The inference in the SoReg model and other MF-based models is simply computing dot-products between the user embedding and item embedding.
%
%
We report the inference cost with regard to the number of items and embedding dimension on \Cref{fig:infer-item} and \Cref{fig:infer-dim} respectively.
For the implementation of CKKS tensor computation, the packing dimension is chosen according to the size of different dimensions, e.g., if the number of items are larger than the embedding dimension, then we pack the same embedding position of different items together.

The results show that, for a small number of items, PADER achieves lower pure computation time compared to CKKS methods. 
Also, the PADER is still more efficient in terms of communication size.
The communication size of PADER is smaller than CKKS methods by at least 100 times in all cases.
In \Cref{fig:limitation}, we further examine scenarios with an exceptionally large number of items, which highlights a limitation of PADER.

\subsection{Benchmark on Real Data}
\begin{table}[h]
\centering
\setlength{\tabcolsep}{5pt}
\begin{tabular}{cccccc}
\toprule
Dataset  & Users    &Items  & Ratings   &  Relations  & Ratio            \\ \midrule
Epinions & 49K      & 139K  & 664K      & 487K          & 90\%             \\
Douban   & 2964     & 40K   & 894K      & 35K           & 10\%             \\
\bottomrule
\end{tabular}
\caption{Dataset details. `Ratio‘ means the portion of data selected for training and the rest is used for testing.}
\label{tab:dataset}
\end{table}

To verify the feasibility of PADER in real-world data, We also conduct experiments on the Epinions dataset~\cite{hamedani2021trust_rec_epinions} and the Douban dataset~\cite{mahao2011social_regularization}.
Dataset details are in \Cref{tab:dataset}.

\subsubsection{Computation and Communication Cost}
To show the efficiency of PADER on real data, we report the pure computation time and communication size of PADER and comparison methods regarding the number of processed users in \Cref{fig:real-benchmark}. For a clearer visual comparison, the pure computation time for the bipartite method without packing is shown only up to 4000 users. 

The results show that in real-world datasets, PADER still maintains high efficiency compared with other methods, both in computation and communication.
\begin{figure}
    \centering
    \includegraphics[width=\linewidth]{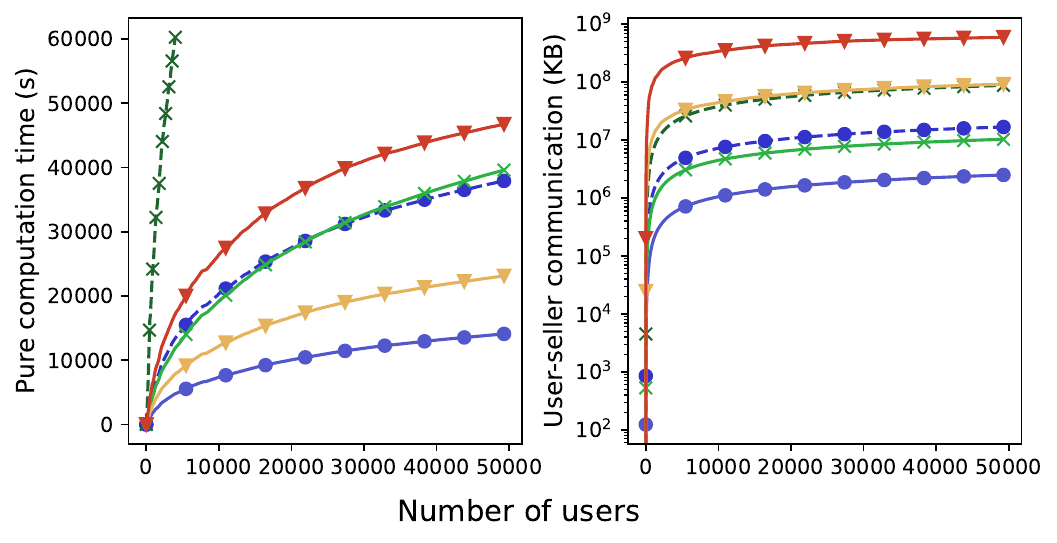}
    \caption{Benchmark on the Epinions dataset.}
    \label{fig:real-benchmark}
\end{figure}

\subsubsection{Model Performance}
\begin{figure}[ht]
    \centering
    \includegraphics[width=\linewidth]{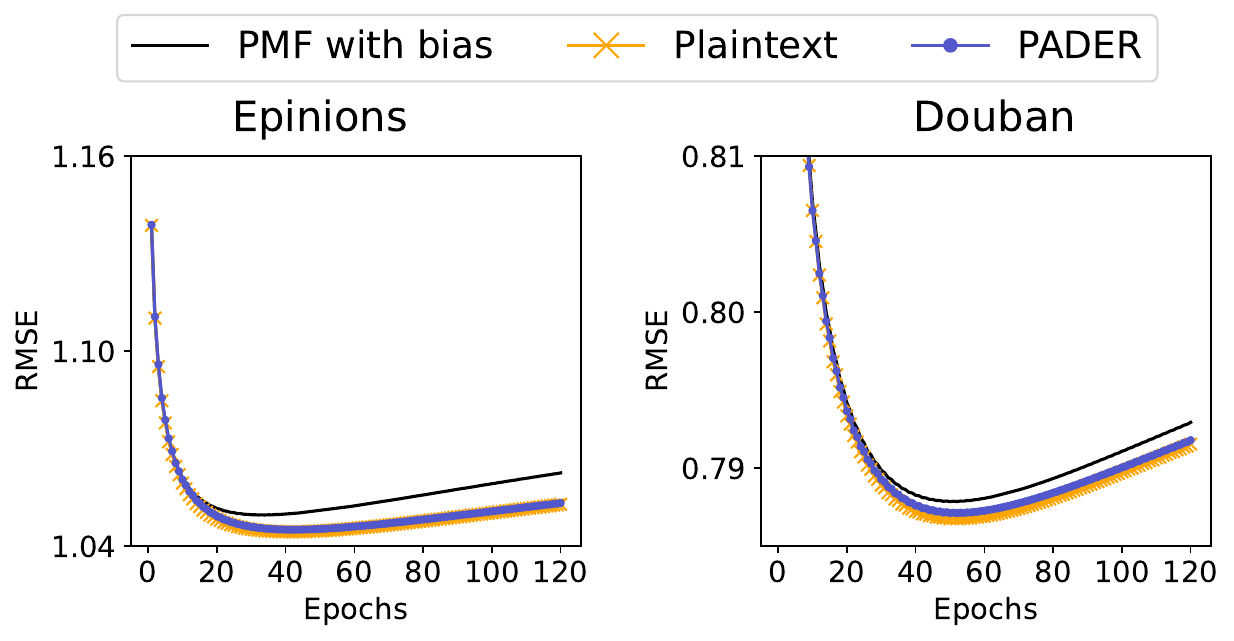}
    \caption{Validation loss in Epinions and Douban.}
    \label{fig:loss-curves}
\end{figure}
We report the validation loss curves of both datasets in \Cref{fig:loss-curves}, and compare PADER with the plaintext SoReg model and PMF model, using the RMSE (Rooted Mean Squared Error) as the metric.
We compare our secure computed SoReg model with the plaintext SoReg model, along with the classical PMF model with bises~\cite{koren2009mf} model.
For both methods, we tune the user, item, and social (if any) regularization coefficients in the range of $\{0, 0.25, 0.5, 1, 2, 4\}$ and report the best result.
The learning rate for SGD is set to 0.003.
During each iteration, up to 8 items and 10 social connections are selected.

We can see that the curves of plaintext SoReg and PADER almost overlap, meaning that our secure computation method does not affect the model accuracy.
Also, both plaintext and secure computation SoReg show improvement in RMSE compared with the biased PMF~\cite{2007pmf,koren2009mf} baseline.

\subsection{Limitation}
In this subsection, we further compared the performance of the PADER method and the CKKS method in the scenario where the number of items is extremely large. The training and inference costs are presented in Figure \ref{fig:limitation}. The pure computation time is demonstrated with the number of items increasing from \(2^1\) to \(2^{14}\), and the communication volume reported is the one when the number of items is equal to \(2^{14}\). 

The results show that, in terms of the pure computation time, the CKKS-tensor method exhibits the best performance. This is because when the number of items is large, it is favorable for the CKKS method as it can pack more plaintexts into one ciphertext for computation. However, when considering the communication size, the proposed PADER method is the most optimal. For example, in terms of the training performance, the communication sizes of PADER and CKKS-tensor are 16.5MB and 95.3MB respectively, with PADER being approximately 50 times smaller than CKKS-tensor.

\subsection{Discussion}
%
In experiments, in terms of communication size, PADER shows superior efficiency than all other methods in most cases, which makes PADER practical in real-life applications where the network bandwidth is limited.
Through comparison between methods with and without packing, we can see that the proposed optimal packing method greatly increases the efficiency by a factor of 8, and the secure computation protocol also overcomes the inefficiency brought by the naive bipartite computation method.

The key advantage of PADER in comparison to CKKS lies in its smaller ciphertext size. In experiments, the communication volume of PADER is typically around 50 to 100 times smaller than that of CKKS. Additionally, when the number of items is relatively small, PADER outperforms CKKS in terms of pure computation time.
\section{Conclusion}
In this paper, we propose PADER, a secure decentralized social recommendation scheme, which is based on the classic social regularization model and Additive Homomorphic Encryption (AHE).
We notice that using a bipartite decomposition can turn the secure training problem into a simple ciphertext-plaintext multiplication which can be trivially done by AHE.
However, simple bipartite computation can be less efficient in certain cases.
To improve the efficiency of the AHE computation, we propose protocols for securely computing polynomials of two parties in arbitrary order, along with an optimal data packing method.
%
Experiment results show that PADER achieves a speedup of $2{\footnotesize \sim}40 \times$ compared with the bipartite computation without data packing, while also being much faster than the state-of-the-art Fully homomorphic encryption method CKKS when considering communication consumption. 
Moreover, the model performance remains the same compared with plaintext training.

\bibliographystyle{IEEEtran}
\bibliography{refs}

\end{document}